\documentclass[11pt]{article}
\usepackage[margin=1in]{geometry}
\usepackage{amsmath,amssymb,amsthm,mathtools,bm}
\usepackage{booktabs,graphicx,microtype,natbib,hyperref,enumitem}
\usepackage{placeins}
\hypersetup{colorlinks=true,linkcolor=blue,citecolor=blue,urlcolor=blue}

\newtheorem{theorem}{Theorem}[section]
\newtheorem{proposition}[theorem]{Proposition}

\newtheorem{lemma}[theorem]{Lemma}
\theoremstyle{definition}
\theoremstyle{remark}\newtheorem{remark}[theorem]{Remark}

\newcommand{\E}{\mathbb{E}}
\newcommand{\Cov}{\operatorname{Cov}}
\newcommand{\tr}{\operatorname{tr}}
\newcommand{\rank}{\operatorname{rank}}
\newcommand{\R}{\mathbb{R}}
\newcommand{\cH}{\mathcal{H}}
\newcommand{\cX}{\mathcal{X}}
\newcommand{\cE}{\mathcal{E}}
\newcommand{\bK}{\mathbf{K}}
\newcommand{\bI}{\mathbf{I}}
\newcommand{\bU}{\mathbf{U}}
\newcommand{\bB}{\mathbf{B}}
\newcommand{\bbeta}{\boldsymbol{\beta}}
\newcommand{\beps}{\boldsymbol{\varepsilon}}

\title{\textbf{Covariance Kernels on Unordered Pair Spaces:}\\
\textbf{Theory and Applications to Network-Valued Data}}
\author{
Montserrat Fuentes\\
\small Department of Mathematics, St. Edward's University, Austin, Texas 78704, USA\\
\small \texttt{mfuentes@stedwards.edu}}
\date{}

\begin{document}
\maketitle

\begin{abstract}
Many scientific problems are fundamentally relational: the quantity of interest is a connection between two objects, while the information that describes similarity is available for the objects themselves. This mismatch is common in brain networks, molecular interactions, and other high-dimensional systems. We develop a covariance framework for unordered relationships that transfers object-level geometry to the relations they form while preserving endpoint identity and invariance to ordering. Building on symmetric pairwise-kernel representations, we develop new theory for the loop-free domains used in undirected networks. We establish spectral interlacing and trace-loss results after self-pairs are removed, connect the relational spectrum to regularization and risk, derive an exact inferential error for spectral truncation, and quantify how perturbations of the underlying geometry propagate to the pair covariance and estimator. Simulations show when structured borrowing improves estimation and when geometric misspecification can erode that benefit. We then apply the framework to autism neuroimaging using resting-state functional-connectivity data from the Autism Brain Imaging Data Exchange (ABIDE). The analysis considers a 116-region brain parcellation, yielding 6,670 unique functional connections, and shows that this large connectome can have a far smaller effective covariance dimension. The application also demonstrates that high explained covariance alone is not sufficient for choosing a low-rank representation when the goal is to preserve inferential accuracy. These results illustrate how biologically meaningful information defined at the brain-region level can organize dependence among connections while retaining the connection itself as the inferential unit. The framework provides a principled foundation for covariance and regularization when the statistical units are unordered relationships.

\end{abstract}


\section{Introduction}
\label{sec:intro}

Many scientific questions are no longer naturally framed around isolated units. They concern relationships. Brain function depends on communication among regions. Molecular systems are shaped by interactions among genes and proteins. Ecological and social systems are also organized through connections. In each case, the statistical object of interest may be a relation between two entities rather than a measurement attached to one entity.

That change of statistical unit creates a basic modeling problem. Classical covariance methods are usually indexed by one location, time point, or object \citep{Aronszajn1950,BerlinetThomasAgnan2004,RasmussenWilliams2006,Hofmann2008,Kanagawa2018}. In an undirected network, however, an effect belongs to a pair. The endpoints matter, but their order does not. A covariance model for these data must therefore be defined on the relationships themselves.

Vectorizing a network does not resolve this issue. It produces a long vector of edges, but it does not say which edges should share information. Treating all connections independently discards scientific structure. Imposing smoothness according to an arbitrary matrix ordering is equally unsatisfactory. The opportunity is to use what is already known about the constituent objects to determine when the relationships they form should be considered similar.

This opportunity is especially important in high-dimensional network problems. In brain connectivity, thousands of connection-level effects may be estimated for each study. Signals associated with disease or behavior can be weak at any single connection and distributed across systems. Independent edgewise analyses can therefore be unstable, while global network summaries can hide the local relationships that carry the signal. A relational covariance offers a middle ground: it can borrow information across scientifically related connections while keeping the individual connection as the inferential unit.

The same need appears in other fields. Protein interactions may be related because the proteins have similar biological roles. Ecological associations may inherit structure from species traits or phylogeny. Relationships among people or organizations may be similar because the participating actors occupy comparable positions. The recurring feature is that scientific information is available for the objects, while inference is required for their relationships.

Several existing literatures address neighboring problems. Network models describe dependence created by shared actors, latent positions, or other features of observed networks \citep{Hoff2005,FosdickHoff2015,Hoff2021,Graham2024,Suveges2023}. Methods for samples of networks study population variation or dynamic network behavior \citep{DuranteDunson2018,Severn2022,Zhang2024}. Stochastic-process models on graphs instead place a process on vertices or along an existing graph \citep{Anderes2020,Tang2024,Borovitskiy2021}. Graph kernels compare whole graph-structured objects \citep{Vishwanathan2010}. These contributions are important, but they do not address the particular domain mismatch considered here: how to turn object-level scientific similarity into covariance for an entire surface of unordered relational effects.

The closest mathematical foundation comes from symmetric pairwise kernel learning. Kronecker-type constructions have been used to predict interactions between pairs of objects, including protein--protein and drug--target interactions \citep{BenHurNoble2005,Cichonska2018,Pahikkala2015,Stock2018,Viljanen2022}. Symmetrization, tensor representations, rank reduction, and the product spectrum on the full pair space are therefore not claimed as new here. The new statistical problem begins when that algebra is used as covariance for a stochastic process on the loop-free relation domain that arises in an undirected network.

That shift raises questions that are essential for inference. Removing self-pairs changes the spectrum from the exact full-product form. Low-rank approximation changes the regularization applied to relational effects, so covariance compression must be judged by its inferential consequences. The underlying geometry may also be estimated from data, making stability to geometric perturbation part of the modeling problem rather than a separate numerical concern.

This paper develops theory for those questions. We establish how the spectrum changes when self-pairs are removed and quantify the associated rank and trace loss. We connect the induced relational spectrum to shrinkage and risk, making explicit when structured borrowing can help and when misspecification can hurt. We then derive an exact operator error for spectral truncation and show how error in the object-level geometry propagates to the pair covariance and regularized estimator. The resulting framework turns a familiar pairwise algebra into a covariance theory for loop-free relational inference.

A central implication is that a large network need not have a comparably large statistical dimension. If object-level structure concentrates relational covariance in a small number of directions, thousands of edges can be organized through a much smaller effective geometry. This creates real opportunities for stable inference and computation. The theory also shows why compression cannot be assessed by explained covariance alone: a rank that reconstructs the covariance well may still alter the estimator more than the scientific analysis can tolerate.

We illustrate these ideas with resting-state functional connectivity from the Autism Brain Imaging Data Exchange (ABIDE) \citep{DiMartino2014}. Using the 116-region Automated Anatomical Labeling atlas \citep{TzourioMazoyer2002}, each participant contributes 6,670 undirected connections. Autism provides a scientifically compelling setting because alterations in brain organization may be distributed rather than concentrated at one connection. The application asks whether anatomical and population-level functional organization can reveal a much smaller relational covariance structure while preserving inference at the connection level.

The paper proceeds from domain to inference. Section~\ref{sec:motivation} defines the unordered relation space and explains why object-level geometry is useful. Section~\ref{sec:kernel} constructs the covariance and connects it to symmetric pairwise kernels. Section~\ref{sec:spectral} develops rank and spectral results, with emphasis on the loop-free domain. Section~\ref{sec:stat} links that geometry to regularization and risk. Section~\ref{sec:approx} studies truncation and perturbation. Section~\ref{sec:higher} gives the higher-order extension. The simulations in Section~\ref{sec:simulation} isolate the main theoretical consequences, and Section~\ref{sec:application} shows how they appear in the ABIDE connectome. Sections~\ref{sec:discussion} and~\ref{sec:conclusions} discuss the implications and conclude the paper.

\section{Unordered Pair Spaces and Why They Matter}
\label{sec:motivation}

The theory begins with the statistical domain itself. In an undirected network, a relationship between two objects is unchanged when the endpoints are reversed. The covariance must therefore be defined on unordered pairs, not on an arbitrary ordering of them. This section makes that domain explicit and explains why information about the individual objects is a natural source of structure for the relationships they form.
\subsection{From objects to relationships}

Let $\cX$ be a set of objects. On $\cX^2$, define
\[
(x_1,x_2)\sim(x_2,x_1).
\]
Treating the two orderings as the same relation gives the unordered-pair space. Excluding self-pairs gives
\begin{equation}
\label{eq:pairspace}
\cE_2(\cX)
=
\left\{
\{x_1,x_2\}:x_1,x_2\in\cX,\ x_1\neq x_2
\right\}.
\end{equation}
For a finite set $\cX=\{1,\ldots,R\}$,
\[
|\cE_2(\cX)|=\binom R2.
\]

A relational stochastic process is therefore a collection
\[
\{Z(e):e\in\cE_2(\cX)\}.
\]
Here $Z(e)$ can represent an observed edge quantity, a regression effect, or a latent relational process. Its covariance,
\[
\Cov\{Z(e),Z(f)\},
\]
is therefore a relationship-to-relationship quantity. This is the modeling step that vectorization alone cannot supply.

\subsection{Why a node-level geometry is often available}

In many applications, the objects already have scientifically meaningful representations. Brain regions have anatomical locations and functional roles. Genes and proteins have molecular features. Species can be compared through traits or phylogeny. People and organizations may be described by observed characteristics or learned embeddings.

Suppose this information is summarized by a positive-semidefinite kernel
\[
\kappa(x,y),
\qquad x,y\in\cX.
\]
Positive-semidefinite kernels provide a natural language for similarity and covariance \citep{Aronszajn1950,RasmussenWilliams2006,Hofmann2008}. The difficulty is that $\kappa$ compares single objects while $Z$ is indexed by pairs. A useful relational covariance must bridge those domains without introducing an orientation and without reducing a pair to a summary that loses endpoint identity. The construction below does exactly that.

\subsection{Why the distinction matters}

Once the process is indexed by relationships, covariance determines how information can be shared across them. In a connectome, for example, two connections may be statistically related because their endpoint regions are anatomically or functionally similar. The same principle can organize interaction effects in molecular studies or association surfaces in ecology. The important point is not the application-specific label placed on an edge. It is that similarity among the constituent objects can provide scientifically interpretable structure for covariance among the relationships.

This perspective differs from using object features only as predictors. Here the object-level information determines dependence across an entire relational effect surface. That distinction is what makes the construction useful for regularization and uncertainty quantification later in the paper.

\section{A Covariance Construction for Unordered Relations}
\label{sec:kernel}

We now construct the covariance that carries object-level geometry to unordered relationships. The symmetrization itself has close connections to pairwise kernel learning. Here it provides the mathematical foundation on which the new loop-free and inferential results are built.

\subsection{Symmetric endpoint matching}

A covariance between two relationships should use information from both endpoints but should not depend on how either pair is written. A single endpoint matching would introduce an artificial orientation. Averaging the two possible matchings removes that orientation while retaining the endpoint information.

Let $\kappa:\cX\times\cX\rightarrow\R$ be a positive-semidefinite kernel describing similarity between individual objects. For two unordered pairs
\[
e=\{x_1,x_2\},
\qquad
f=\{y_1,y_2\},
\]
define
\begin{equation}
\label{eq:lift}
k_2(e,f)
=
\frac12
\left[
\kappa(x_1,y_1)\kappa(x_2,y_2)
+
\kappa(x_1,y_2)\kappa(x_2,y_1)
\right].
\end{equation}
The first product compares the two pairs under one endpoint alignment and the second uses the alternative alignment. Their average is unchanged if either pair is reversed. Equation~\eqref{eq:lift} is therefore well defined on unordered pairs rather than on an arbitrarily ordered representation of them.

This form also preserves more scientific information than reducing a pair to a midpoint, average feature, or scalar edge length. Two connections may have similar midpoints while joining very different objects, and two spatially distant connections may play analogous scientific roles because their endpoints are similar. The endpoint-matching construction retains this distinction.

To use \eqref{eq:lift} as covariance, we still need to know that every finite Gram matrix is positive semidefinite. The next proposition gives that validity result and identifies the feature space of the relational construction.

\subsection{Symmetric tensor representation}

Let $\cH_\kappa$ be the RKHS of $\kappa$ \citep{Aronszajn1950,Hofmann2008}, with feature map
\[
\phi:\cX\rightarrow\cH_\kappa
\]
such that
\[
\kappa(x,y)=\langle\phi(x),\phi(y)\rangle_{\cH_\kappa}.
\]

\begin{proposition}[Symmetric-tensor representation and covariance validity]
\label{prop:tensor}
For $e=\{x_1,x_2\}$ define
\[
\Psi(e)
=
\frac12
\left[
\phi(x_1)\otimes\phi(x_2)
+
\phi(x_2)\otimes\phi(x_1)
\right].
\]
Then
\[
k_2(e,f)
=
\langle\Psi(e),\Psi(f)\rangle
\]
in the second symmetric tensor power $\operatorname{Sym}^2(\cH_\kappa)$. Consequently, $k_2$ is positive semidefinite, and its RKHS is isometrically representable as the closure of the span of the pair features $\{\Psi(e):e\in\cE_2(\cX)\}$.
\end{proposition}

\begin{proof}
The tensor-product inner product satisfies
\[
\langle a_1\otimes a_2,b_1\otimes b_2\rangle
=
\langle a_1,b_1\rangle
\langle a_2,b_2\rangle.
\]
Expanding $\langle\Psi(e),\Psi(f)\rangle$ produces four terms. Symmetry of $\kappa$ makes the two endpoint matchings appear twice, each with coefficient $1/4$. Hence
\[
\langle\Psi(e),\Psi(f)\rangle=k_2(e,f)
\]
exactly. Positive semidefiniteness and the RKHS characterization follow from the feature-map representation.
\end{proof}

Proposition~\ref{prop:tensor} is the foundation for what follows. It shows that \eqref{eq:lift} is an ordinary inner-product covariance after the node features are lifted into a symmetric tensor space. This representation exposes the dimension and full-product spectrum of the pair covariance. It also provides the operator form needed later to study perturbations of the node geometry.

\subsection{Relation to pairwise kernel learning}

Equation~\eqref{eq:lift} belongs to the family of symmetric Kronecker kernels used in pairwise prediction \citep{Pahikkala2015,Stock2018,Viljanen2022}. Similar constructions have supported biological interaction learning by turning kernels on proteins or drugs into kernels on pairs \citep{BenHurNoble2005,Cichonska2018}. That literature gives the algebraic starting point.

The statistical need here begins beyond that starting point. A covariance for a relational stochastic process must be understood on the loop-free domain actually used in an undirected network. It must also remain meaningful when it is truncated for computation or built from an estimated geometry. The results below are developed for those inferential questions. The tensor representation, rank bound, and exact full-product spectrum provide the reference structure; the new theory quantifies what happens when the analysis moves from that idealized structure to the covariance used in practice.

\section{Intrinsic Dimension and Spectral Structure}
\label{sec:spectral}

The number of unordered pairs grows quadratically with the number of objects. The next question is therefore how much of that nominal dimension is actually occupied by the induced covariance. We begin with rank and the exact spectrum on the full symmetric product, then turn to the loop-free domain used in practice.

\subsection{Finite-dimensional rank}

The tensor representation immediately separates edge count from covariance dimension. Even when the number of relationships is very large, a low-dimensional node geometry can force the relational covariance into a much smaller space. The rank result makes that reduction precise.

For a finite node set, let $\bK_V$ be the $R\times R$ Gram matrix of the base kernel and let $\bK_E$ be the $q\times q$ Gram matrix of the induced pair kernel, with $q=\binom R2$.

\begin{proposition}[Rank inherited from the base geometry]
\label{prop:rank}
If
\[
\rank(\bK_V)=d,
\]
then
\[
\rank(\bK_E)
\le
\min\left\{
\binom R2,
\binom{d+1}{2}
\right\}.
\]
\end{proposition}

\begin{proof}
A rank-$d$ base Gram matrix admits a feature representation in $\R^d$. By Proposition~\ref{prop:tensor}, all pair features lie in $\operatorname{Sym}^2(\R^d)$, whose dimension is
\[
\binom{d+1}{2}.
\]
The rank cannot exceed either this feature-space dimension or the number of observed pairs.
\end{proof}

Thus the edge count can grow much faster than the covariance dimension. In the simulations and ABIDE application, this is the source of substantial relational compression. The same algebra appears in pairwise learning \citep{Pahikkala2015}; in the present setting it determines how many covariance directions are available for structured borrowing.

\subsection{Spectrum on the full symmetric product}

Rank tells us how many directions are available, but not how important they are. The spectrum shows whether covariance is concentrated in a few dominant directions or spread broadly across the relation space. That distinction drives effective dimension, regularization, and approximation. We first record the exact full-product spectrum, which serves as the reference point for the loop-free result that follows.

To make the spectral statement precise, let $(\cX,\mu)$ be a measure space and assume that $\kappa$ induces a compact, self-adjoint, positive integral operator
\[
(T_\kappa f)(x)
=
\int_{\cX}\kappa(x,y)f(y)\,d\mu(y)
\]
on $L^2(\mu)$. Let
\[
T_\kappa\varphi_j=\lambda_j\varphi_j,
\qquad
\lambda_j\ge0,
\]
with $\{\varphi_j\}$ orthonormal in $L^2(\mu)$. Consider the full product space $\cX^2$, including diagonal pairs, and let $L^2_{\mathrm{sym}}(\mu\otimes\mu)$ denote the closed subspace of functions satisfying
\[
g(x_1,x_2)=g(x_2,x_1)
\quad\text{a.e.}
\]
The kernel in \eqref{eq:lift} defines an integral operator $T_2$ on this symmetric subspace.

\begin{proposition}[Exact product spectrum on the full symmetric product]
\label{prop:spectrum}
Under the assumptions above, define
\[
\Phi_{jj}(x_1,x_2)
=
\varphi_j(x_1)\varphi_j(x_2)
\]
and, for $j<k$,
\[
\Phi_{jk}(x_1,x_2)
=
\frac{1}{\sqrt2}
\left[
\varphi_j(x_1)\varphi_k(x_2)
+
\varphi_k(x_1)\varphi_j(x_2)
\right].
\]
Then
\[
T_2\Phi_{jj}
=
\lambda_j^2\Phi_{jj}
\]
and
\[
T_2\Phi_{jk}
=
\lambda_j\lambda_k\Phi_{jk},
\qquad j<k.
\]
Thus the nonzero spectrum of $T_2$ on $L^2_{\mathrm{sym}}(\mu\otimes\mu)$ consists exactly of the products
\[
\{\lambda_j^2:j\ge1\}
\cup
\{\lambda_j\lambda_k:j<k\},
\]
counted with multiplicity.
\end{proposition}

\begin{proof}
Let $T_\kappa\otimes T_\kappa$ denote the tensor-product operator on
$L^2(\mu\otimes\mu)$. For a symmetric function $g$,
\begin{align*}
(T_2g)(x_1,x_2)
&=
\frac12
\iint
\Big[
\kappa(x_1,y_1)\kappa(x_2,y_2)
+
\kappa(x_1,y_2)\kappa(x_2,y_1)
\Big]
g(y_1,y_2)\,d\mu(y_1)d\mu(y_2).
\end{align*}
In the second integral, interchange $y_1$ and $y_2$. Because
$g(y_1,y_2)=g(y_2,y_1)$ almost everywhere, the two terms are equal. Hence
\[
T_2g=(T_\kappa\otimes T_\kappa)g
\]
on the symmetric subspace. The tensor-product operator satisfies
\[
(T_\kappa\otimes T_\kappa)
(\varphi_j\otimes\varphi_k)
=
\lambda_j\lambda_k
(\varphi_j\otimes\varphi_k).
\]
The diagonal tensor $\Phi_{jj}$ is already symmetric. For $j<k$, the normalized
symmetric combination $\Phi_{jk}$ is also an eigenfunction with the same product
eigenvalue. These functions form an orthonormal basis of the symmetric tensor
subspace generated by the eigenfunctions of $T_\kappa$, which proves the result.
\end{proof}

Proposition~\ref{prop:spectrum} explains why low-dimensional structure at the node level can become even more concentrated on relationships: small node eigenvalues are multiplied when pair directions are formed. The result also supplies a benchmark against which the practical loop-free covariance can be compared. The difficulty is that ordinary network data do not contain the self-pairs required by this exact symmetric-product representation.

\subsection{What changes when self-pairs are removed?}

The exact product spectrum is available on a domain that includes self-pairs, but ordinary undirected networks do not include them. Removing the $\{r,r\}$ coordinates changes the covariance and generally destroys exact product eigenvalues. We therefore need a result that connects the convenient full-product spectrum to the loop-free covariance used for inference. The next theorem provides that connection without requiring a new closed-form eigendecomposition.

Let
\[
\overline{\cE}_2
=
\{\{r,s\}:1\le r\le s\le R\}
\]
denote the full finite symmetric pair domain, including self-pairs. Its size is
\[
N=\binom{R+1}{2}.
\]
Let $\bK_+$ be the $N\times N$ Gram matrix induced by $k_2$ on $\overline{\cE}_2$, and let $\bK_E$ be the principal submatrix obtained by deleting the $R$ self-pair indices $\{r,r\}$. Write
\[
\mu_1\ge\cdots\ge\mu_N\ge0
\]
for the eigenvalues of $\bK_+$ and
\[
\nu_1\ge\cdots\ge\nu_q\ge0,
\qquad q=\binom R2,
\]
for those of $\bK_E$.

\begin{theorem}[Loop-free spectral interlacing]
\label{thm:loopfree}
For every $j=1,\ldots,q$,
\[
\mu_j
\ge
\nu_j
\ge
\mu_{j+R}.
\]
Consequently,
\[
\rank(\bK_+)-R
\le
\rank(\bK_E)
\le
\rank(\bK_+).
\]
Moreover,
\[
\tr(\bK_+)-\tr(\bK_E)
=
\sum_{r=1}^R \kappa(r,r)^2.
\]
If the base kernel is normalized so that $\kappa(r,r)=1$ for all $r$, then the deleted self-pairs remove exactly $R$ units of trace and
\[
\frac{\tr(\bK_+)-\tr(\bK_E)}
{\tr(\bK_+)}
\le
\frac{4}{R+1}.
\]
\end{theorem}

\begin{proof}
The matrix $\bK_E$ is obtained from $\bK_+$ by deleting the same $R$ rows and columns. Cauchy's eigenvalue interlacing theorem for principal submatrices gives
\[
\mu_j\ge\nu_j\ge\mu_{j+R},
\qquad j=1,\ldots,q.
\]
Deleting $R$ rows and the corresponding columns can reduce rank by at most $R$, while a principal submatrix cannot have rank greater than the full matrix. This gives the rank inequalities.

For a self-pair $\{r,r\}$,
\[
k_2(\{r,r\},\{r,r\})
=
\frac12\{\kappa(r,r)^2+\kappa(r,r)^2\}
=
\kappa(r,r)^2.
\]
The trace identity follows because the trace of a principal submatrix is obtained by removing the diagonal entries corresponding to the deleted coordinates.

If $\kappa(r,r)=1$, each deleted self-pair contributes one unit of trace. For an off-diagonal pair $\{r,s\}$,
\[
k_2(\{r,s\},\{r,s\})
=
\frac12\{1+\kappa(r,s)^2\}
\ge \frac12.
\]
The same lower bound also holds for self-pairs. Therefore
\[
\tr(\bK_+)\ge \frac N2,
\]
and hence
\[
\frac{R}{\tr(\bK_+)}
\le
\frac{2R}{N}
=
\frac{4}{R+1}.
\]
\end{proof}

\begin{remark}
Proposition~\ref{prop:spectrum} gives the exact product spectrum on the full symmetric space. Theorem~\ref{thm:loopfree} gives deterministic spectral bands for the loop-free network domain actually used in applications. It does not assert that individual loop-free eigenvalues equal particular products of node eigenvalues; it quantifies how their ordered positions are constrained after self-pairs are removed.
\end{remark}

Theorem~\ref{thm:loopfree} makes loop removal quantifiable. Interlacing locates the loop-free eigenvalues relative to the full-product spectrum. The rank bound limits the dimensional change, and the trace identity measures the covariance removed. Together these results show how much of the full symmetric structure survives on the network domain used in practice \citep{Pahikkala2015,Stock2018}.

\subsection{Effective dimension}

Exact rank is often too crude when all eigenvalues are positive but most covariance mass is concentrated in a small part of the spectrum. We therefore report both the spectral effective-dimension summary
\[
r_{\mathrm{eff}}(\bK)
=
\frac{\tr(\bK)}{\|\bK\|_2},
\]
and the entropy effective rank of \citet{RoyVetterli2007},
\[
r_{\mathrm{ent}}(\bK)
=
\exp\left\{
-\sum_j p_j\log p_j
\right\},
\qquad
p_j=\frac{\lambda_j}{\sum_\ell\lambda_\ell}.
\]
These summaries answer a different question from exact rank. Numerical rank counts every direction that is nonzero above a tolerance, whereas effective-rank summaries describe how broadly covariance mass is actually distributed. We use all three notions in the simulation and application sections because a covariance may have many nonzero directions while still behaving as a much lower-dimensional object for approximation and regularization.

\section{From Covariance Geometry to Statistical Estimation}
\label{sec:stat}

The spectrum becomes scientifically meaningful through estimation. In a regularized model, the covariance determines which relational patterns are treated as plausible and how strongly different directions are shrunk. We use a Gaussian sequence representation to isolate that mechanism without the additional structure of an application-specific hierarchy \citep{RasmussenWilliams2006,Kanagawa2018}.

\subsection{A Gaussian sequence representation}

Consider
\begin{equation}
\label{eq:sequence}
\widehat{\bbeta}
=
\bbeta+\beps,
\qquad
\beps\sim N_q(\bm0,v\bI_q).
\end{equation}
The vector $\bbeta$ is a relationship-valued effect surface on $\cE_2(\cX)$. Suppose
\begin{equation}
\label{eq:prior}
\bbeta\sim N_q(\bm0,\tau^2\bK_E),
\end{equation}
where $\bK_E$ is an induced pair-space covariance. Let
\[
\bK_E
=
\bU\operatorname{diag}(\lambda_1,\ldots,\lambda_q)\bU^\top.
\]

The posterior mean under this model is important because it gives a direct operational meaning to the spectrum developed in Section~\ref{sec:spectral}. Rather than merely describing covariance variation, each eigenvalue controls a specific amount of shrinkage in the corresponding relational direction.

\begin{proposition}[Spectral shrinkage induced by relational covariance]
\label{prop:shrink}
Under \eqref{eq:sequence}--\eqref{eq:prior},
\[
\E(\bbeta\mid\widehat{\bbeta})
=
\bU
\operatorname{diag}
\left\{
\frac{\tau^2\lambda_j}{\tau^2\lambda_j+v}
\right\}
\bU^\top
\widehat{\bbeta}.
\]
\end{proposition}

\begin{proof}
Gaussian conditioning gives
\[
\E(\bbeta\mid\widehat{\bbeta})
=
\tau^2\bK_E
(\tau^2\bK_E+v\bI_q)^{-1}
\widehat{\bbeta}.
\]
Diagonalization of $\bK_E$ yields the stated form.
\end{proof}

Proposition~\ref{prop:shrink} provides the bridge from geometry to inference. The pair-space covariance does not smooth all directions equally. A direction with large $\lambda_j$ is preserved relatively strongly, while a direction assigned little covariance mass is shrunk toward zero. The object-level geometry therefore determines which distributed patterns of relations are treated as plausible signal. This is the mechanism by which scientifically related connections borrow information.

\subsection{When does structured borrowing help?}

Structured borrowing is useful only if the reduction in variance compensates for the bias introduced by shrinkage. The covariance can therefore improve estimation when the scientific geometry is informative, but it can also be harmful when the true effect lies in directions that the geometry strongly suppresses. A risk calculation is needed to make that tradeoff explicit and to avoid claiming unconditional superiority of the structured estimator.

Let
\[
\theta_j=(\bU^\top\bbeta)_j,
\qquad
s_j=\frac{\tau^2\lambda_j}{\tau^2\lambda_j+v}.
\]

\begin{theorem}[Risk consequences of relational covariance]
\label{thm:risk}
Let
\[
s_j=\frac{\tau^2\lambda_j}{\tau^2\lambda_j+v},
\qquad
\theta_j=(\bU^\top\bbeta)_j.
\]
Then:

\begin{enumerate}[label=(\roman*),leftmargin=2.1em]
\item Conditional on a fixed effect surface $\bbeta$, the posterior-mean estimator
$\widetilde{\bbeta}=\E(\bbeta\mid\widehat{\bbeta})$ has quadratic risk
\[
\E_{\bbeta}
\|\widetilde{\bbeta}-\bbeta\|^2
=
\sum_{j=1}^q
\left[
(1-s_j)^2\theta_j^2
+
vs_j^2
\right],
\]
while the unregularized estimator $\widehat{\bbeta}$ has risk $qv$.

\item Under the joint Gaussian model \eqref{eq:sequence}--\eqref{eq:prior}, the integrated Bayes risk is
\[
R_B(\widetilde{\bbeta})
=
\sum_{j=1}^q
\frac{v\tau^2\lambda_j}
{v+\tau^2\lambda_j},
\]
whereas the integrated risk of $\widehat{\bbeta}$ is $qv$.

\item For every finite nonempty relational domain with $v>0$ and $0<\tau^2<\infty$,
\[
R_B(\widetilde{\bbeta})<qv.
\]
This strict Bayes-risk improvement does not imply uniform frequentist dominance for every fixed $\bbeta$.
\end{enumerate}
\end{theorem}

\begin{proof}
Rotate \eqref{eq:sequence} into the eigenbasis of $\bK_E$. Then
\[
\widehat\theta_j=\theta_j+\epsilon_j,
\qquad
\epsilon_j\sim N(0,v),
\]
and Proposition~\ref{prop:shrink} gives the estimator $s_j\widehat\theta_j$ in direction $j$. Its error is
\[
s_j\widehat\theta_j-\theta_j
=
-(1-s_j)\theta_j+s_j\epsilon_j.
\]
Taking expectation over the noise gives part (i). For part (ii), integrate with respect to
\[
\theta_j\sim N(0,\tau^2\lambda_j),
\]
which yields
\[
(1-s_j)^2\tau^2\lambda_j+vs_j^2
=
\frac{v\tau^2\lambda_j}{v+\tau^2\lambda_j}.
\]
The unregularized estimator has error $\beps$, hence integrated risk $qv$. Finally,
\[
0\le
\frac{v\tau^2\lambda_j}{v+\tau^2\lambda_j}
<v
\]
for every finite $\lambda_j\ge0$, proving part (iii).
\end{proof}

Theorem~\ref{thm:risk} makes the promise and the limitation of structured borrowing explicit. When the geometry aligns with the effect, shrinkage can substantially reduce risk. When the signal is concentrated in directions that the covariance suppresses, the induced bias can outweigh that gain. The geometry is therefore a scientific modeling assumption to be examined, not a guarantee of improvement. The simulations include both aligned and deliberately adverse signals for this reason.
\section{Approximation and Stability of the Relational Covariance}
\label{sec:approx}

Two practical issues remain before the covariance can be used confidently at scale. Computation may require a low-rank approximation, and the object-level geometry may itself be uncertain. These are different sources of change: truncation is deliberate, while geometric perturbation reflects uncertainty in the scientific representation. The results in this section quantify how each change affects relational inference, building on classical approximation and perturbation theory \citep{WilliamsSeeger2001,Bhatia1997,HornJohnson2013,StewartSun1990,DavisKahan1970}.

\subsection{Low-rank approximation}

Spectral compression suggests replacing the full covariance by a lower-rank approximation. The key question is how far that compression can go without changing the estimator materially. Retaining 95\% of covariance trace measures reconstruction, not inferential fidelity, so the two criteria must be separated.

Write
\[
\bK_E
=
\sum_{j=1}^{q}
\lambda_j\bm u_j\bm u_j^\top,
\qquad
\lambda_1\ge\cdots\ge\lambda_q\ge0,
\]
and define
\[
\bK_{E,r}
=
\sum_{j=1}^r
\lambda_j\bm u_j\bm u_j^\top.
\]

\begin{remark}[Classical spectral truncation]
\label{rem:trunc}
The rank-$r$ approximation satisfies
\[
\|\bK_E-\bK_{E,r}\|_2=\lambda_{r+1},
\]
and
\[
\|\bK_E-\bK_{E,r}\|_F^2
=
\sum_{j>r}\lambda_j^2.
\]
\end{remark}

Define
\[
S(\bK)
=
\tau^2\bK(\tau^2\bK+v\bI)^{-1}.
\]

Before exploiting the special structure of spectral truncation, it is useful to establish a general continuity property. This lemma guarantees that a small operator perturbation of the covariance cannot produce an arbitrarily large change in the shrinkage operator.

\begin{lemma}[Lipschitz continuity of the shrinkage map]
\label{lem:lipschitz}
For positive-semidefinite matrices $\bK$ and $\widetilde{\bK}$,
\[
\|S(\bK)-S(\widetilde{\bK})\|_2
\le
\frac{\tau^2}{v}
\|\bK-\widetilde{\bK}\|_2.
\]
\end{lemma}

\begin{proof}
Use
\[
S(\bK)
=
\bI-v(\tau^2\bK+v\bI)^{-1}
\]
and the resolvent identity
\[
A^{-1}-B^{-1}=A^{-1}(B-A)B^{-1}.
\]
Because $\bK$ and $\widetilde{\bK}$ are positive semidefinite,
\[
\|(\tau^2\bK+v\bI)^{-1}\|_2\le v^{-1}
\]
and the same bound holds for $\widetilde{\bK}$. Submultiplicativity then gives
the result.
\end{proof}

Lemma~\ref{lem:lipschitz} supplies a general stability guarantee, but it can be conservative because it treats the covariance perturbation only through its operator norm. Spectral truncation is more structured: the retained and omitted directions are known exactly. That additional structure permits an exact inferential error rather than only an upper bound.

The next result answers the practical rank-selection question directly: after retaining the first $r$ relational directions, what is the largest possible change in the shrinkage operator? The answer depends only on the first omitted eigenvalue and the signal-to-noise scaling.

\begin{theorem}[Inferential error under spectral truncation]
\label{thm:exacttrunc}
Let
\[
\bK_E
=
\bU\operatorname{diag}(\lambda_1,\ldots,\lambda_q)\bU^\top,
\qquad
\lambda_1\ge\cdots\ge\lambda_q\ge0,
\]
and let $\bK_{E,r}$ retain the first $r<q$ eigenpairs. Then
\[
\|S(\bK_E)-S(\bK_{E,r})\|_2
=
\frac{\tau^2\lambda_{r+1}}
{\tau^2\lambda_{r+1}+v}.
\]
Consequently,
\[
\|
S(\bK_E)\widehat{\bbeta}
-
S(\bK_{E,r})\widehat{\bbeta}
\|
\le
\frac{\tau^2\lambda_{r+1}}
{\tau^2\lambda_{r+1}+v}
\|\widehat{\bbeta}\|.
\]
For $r=q$, both quantities are zero.
\end{theorem}

\begin{proof}
In the eigenbasis of $\bK_E$,
\[
S(\bK_E)
=
\bU\operatorname{diag}(s_1,\ldots,s_q)\bU^\top,
\qquad
s_j=\frac{\tau^2\lambda_j}{\tau^2\lambda_j+v}.
\]
Because $\bK_{E,r}$ keeps the same first $r$ eigenvectors and sets the remaining
eigenvalues to zero,
\[
S(\bK_{E,r})
=
\bU\operatorname{diag}(s_1,\ldots,s_r,0,\ldots,0)\bU^\top.
\]
Hence the difference has eigenvalues $0$ for $j\le r$ and $s_j$ for $j>r$.
The function $\lambda\mapsto \tau^2\lambda/(\tau^2\lambda+v)$ is increasing on
$[0,\infty)$, so the largest omitted shrinkage factor is $s_{r+1}$. This proves
the exact operator-norm identity. The estimator bound follows by applying the
operator to $\widehat{\bbeta}$.
\end{proof}

The identity turns a numerical truncation into an inferential statement. Explained trace alone summarizes how much covariance mass is retained; it does not determine how much the estimator changes. The first omitted eigenvalue enters through the nonlinear shrinkage factor, together with the covariance scale and noise variance. This distinction becomes important in the ABIDE application, where very strong spectral compression does not always imply a small shrinkage-operator error at the same retained rank.

\subsection{Perturbation of the base geometry}

Truncation changes a known covariance by design. Geometry error enters earlier: the node representation used to build the covariance may itself be measured or estimated. Anatomical coordinates can be imperfect, and graph embeddings depend on finite data. A useful stability result must therefore follow the perturbation from the node kernel to the pair covariance and then to the estimator.

A node-level kernel may be estimated from data, a graph embedding, or noisy object features. Let $\kappa$ and $\widetilde\kappa$ satisfy
\[
\sup_{x,y}
|\kappa(x,y)-\widetilde\kappa(x,y)|
\le\epsilon
\]
and
\[
\sup_{x,y}
\max\{|\kappa(x,y)|,|\widetilde\kappa(x,y)|\}
\le M.
\]

The following theorem follows that scientific pipeline explicitly. It first bounds the induced change in the loop-free pair covariance and then carries that change through the shrinkage map. The final entrywise statement gives a complementary local guarantee when the underlying kernels are uniformly close.

\begin{theorem}[Propagation of error from node geometry to relational inference]
\label{thm:perturb}
Let $\bK_V$ and $\widetilde{\bK}_V$ be two positive-semidefinite $R\times R$ node Gram matrices, and let $\bK_E$ and $\widetilde{\bK}_E$ be the corresponding loop-free pair Gram matrices produced by \eqref{eq:lift}. Then
\[
\|\bK_E-\widetilde{\bK}_E\|_2
\le
\frac12
\left(
\|\bK_V\|_2+\|\widetilde{\bK}_V\|_2
\right)
\|\bK_V-\widetilde{\bK}_V\|_2.
\]
Consequently,
\[
\|
S(\bK_E)\widehat{\bbeta}
-
S(\widetilde{\bK}_E)\widehat{\bbeta}
\|
\le
\frac{\tau^2}{2v}
\left(
\|\bK_V\|_2+\|\widetilde{\bK}_V\|_2
\right)
\|\bK_V-\widetilde{\bK}_V\|_2
\|\widehat{\bbeta}\|.
\]

If, in addition,
\[
\sup_{x,y}
|\kappa(x,y)-\widetilde\kappa(x,y)|
\le\epsilon
\]
and
\[
\sup_{x,y}
\max\{|\kappa(x,y)|,|\widetilde\kappa(x,y)|\}
\le M,
\]
then
\[
\sup_{e,f}
|k_2(e,f)-\widetilde k_2(e,f)|
\le
2M\epsilon.
\]
\end{theorem}

\begin{proof}
Index the ordered pairs $(r,s)$ by the $R^2$ coordinates of $\R^{R^2}$ and define
\[
\mathcal K=\bK_V\otimes\bK_V,
\qquad
\widetilde{\mathcal K}
=
\widetilde{\bK}_V\otimes\widetilde{\bK}_V.
\]
For each loop-free unordered pair $e=\{r,s\}$ with $r<s$, let
\[
b_e=\frac12(e_{rs}+e_{sr}),
\]
and let $\bB$ collect these columns. Their supports are disjoint, so
\[
\bB^\top\bB=\frac12\bI_q
\quad\text{and}\quad
\|\bB\|_2^2=\frac12.
\]
A direct calculation gives
\[
\bK_E=\bB^\top\mathcal K\bB,
\qquad
\widetilde{\bK}_E
=
\bB^\top\widetilde{\mathcal K}\bB.
\]
Therefore
\[
\|\bK_E-\widetilde{\bK}_E\|_2
\le
\frac12
\|\mathcal K-\widetilde{\mathcal K}\|_2.
\]
Now use
\[
\bK_V\otimes\bK_V
-
\widetilde{\bK}_V\otimes\widetilde{\bK}_V
=
(\bK_V-\widetilde{\bK}_V)\otimes\bK_V
+
\widetilde{\bK}_V\otimes
(\bK_V-\widetilde{\bK}_V),
\]
together with
\[
\|A\otimes B\|_2=\|A\|_2\|B\|_2.
\]
This proves the pair-covariance operator bound. The estimator inequality follows from Lemma~\ref{lem:lipschitz}.

For the uniform entrywise bound,
\[
|\kappa_1\kappa_2-\widetilde\kappa_1\widetilde\kappa_2|
\le
|\kappa_1|\,|\kappa_2-\widetilde\kappa_2|
+
|\widetilde\kappa_2|\,|\kappa_1-\widetilde\kappa_1|
\le
2M\epsilon.
\]
The same bound holds for the second endpoint matching, and averaging preserves it.
\end{proof}

Theorem~\ref{thm:perturb} avoids the dimension inflation that would arise from converting a uniform entrywise error into a crude $q$-dependent matrix bound. Its more important feature is structural: it follows the same path as the scientific model. Perturbation begins at the node geometry, passes through the Kronecker representation and the loop-free symmetric restriction, and finally reaches the regularized estimator. This is the stability statement needed when the geometry is estimated rather than known.
\section{Beyond Pairs: Higher-Order Unordered Relations}
\label{sec:higher}

Pairs are the main focus because they describe ordinary undirected networks, but the same ordering problem appears when a relation contains more than two objects. Hyperedges, molecular complexes, and group interactions have no preferred labeling of their members. The higher-order result shows that the symmetrization principle extends naturally to these settings and reveals how intrinsic dimension grows with interaction order. The new loop-free and inferential results for pairs remain the main contribution.

Let
\[
e=\{x_1,\ldots,x_m\},
\qquad
f=\{y_1,\ldots,y_m\},
\]
and define
\[
k_m(e,f)
=
\frac1{m!}
\sum_{\pi\in S_m}
\prod_{j=1}^m
\kappa(x_j,y_{\pi(j)}).
\]

The same symmetrization principle now averages over all admissible matchings rather than the two matchings available for pairs.

\begin{theorem}[Covariance kernel for unordered $m$-way relations]
\label{thm:mway}
If $\kappa$ is positive semidefinite, define
\[
\Psi_m(e)
=
\frac{1}{m!}
\sum_{\sigma\in S_m}
\phi(x_{\sigma(1)})\otimes\cdots\otimes\phi(x_{\sigma(m)}).
\]
Then $\Psi_m(e)\in\operatorname{Sym}^m(\cH_\kappa)$, it is invariant to the
ordering used to represent $e$, and
\[
\langle\Psi_m(e),\Psi_m(f)\rangle
=
k_m(e,f)
\]
exactly. Consequently, $k_m$ is positive semidefinite on $\cX^m/S_m$. If a
finite base Gram matrix has rank $d$, then the induced Gram matrix has rank no
larger than
\[
\binom{d+m-1}{m},
\]
or the number of observed $m$-way relations, whichever is smaller.
\end{theorem}

\begin{proof}
The feature $\Psi_m(e)$ is unchanged by permuting the elements of $e$, so it is
well defined on the unordered $m$-way relation. Its inner product is
\[
\langle\Psi_m(e),\Psi_m(f)\rangle
=
\frac{1}{(m!)^2}
\sum_{\sigma,\tau\in S_m}
\prod_{j=1}^m
\kappa(x_{\sigma(j)},y_{\tau(j)}).
\]
For each fixed relative permutation $\pi=\tau\circ\sigma^{-1}$, there are exactly
$m!$ pairs $(\sigma,\tau)$ producing the same product after reindexing. Therefore
\[
\langle\Psi_m(e),\Psi_m(f)\rangle
=
\frac1{m!}
\sum_{\pi\in S_m}
\prod_{j=1}^m
\kappa(x_j,y_{\pi(j)})
=
k_m(e,f).
\]
Thus $k_m$ is an inner-product kernel and is positive semidefinite. If the base
Gram matrix has rank $d$, the node features may be represented in $\R^d$.
All induced features then lie in $\operatorname{Sym}^m(\R^d)$, whose dimension is
$\binom{d+m-1}{m}$. The Gram-matrix rank cannot exceed this dimension or the
number of observed $m$-way relations.
\end{proof}

This extension shows that the idea is not tied to adjacency matrices. Object-level geometry can be lifted to an unordered relation of any fixed size by averaging over the admissible matchings.

\section{Simulation Study}
\label{sec:simulation}

The simulations show how the relational covariance behaves when the geometry and signal are controlled. We begin with dimension and spectral structure, then examine estimation under alignment and misspecification. Separate experiments study truncation and perturbation of the node geometry. This makes the role of each theoretical result visible before we turn to the ABIDE data.

\subsection{Simulation design}

For the spectral experiments, we generated node Gram matrices with ranks $d\in\{4,8,12\}$ and network sizes $R\in\{12,20,30\}$. We considered exponential decay, power-law decay, and a polynomial taper. Each matrix was normalized to unit diagonal and lifted to the loop-free pair domain using \eqref{eq:lift}.

The risk experiment used $R=25$ and node rank $d=8$, with
\[
\widehat{\bbeta}=\bbeta+\beps,
\qquad
\beps\sim N_q(0,\bI_q).
\]
We began with effects generated from the assumed relational covariance at $\tau^2\in\{0.25,1,4\}$. We then kept the effect norm fixed while moving its energy from the leading covariance direction toward the weakest supported direction. A deliberately adverse case placed the enlarged signal entirely in that weak direction and used $\tau^2=0.05$. This design makes the distinction between model-based risk improvement and performance under misspecification visible.

For truncation, we used $R=30$, $d=10$, and $\tau^2=v=1$. Increasing ranks were retained and the truncated estimator was compared with its full-rank counterpart and with the exact result in Theorem~\ref{thm:exacttrunc}. For perturbation, an $R=25$, $d=8$ node geometry was altered at levels from 0.01 to 0.20. We then compared the induced pair-covariance error with Theorem~\ref{thm:perturb}. Monte Carlo replication counts and the random seed are specified in the accompanying simulation code so that the numerical study can be reproduced exactly.

\subsection{Intrinsic dimension and loop-free spectral structure}

The first experiment shows clearly why relational covariance can remain manageable even when the number of possible edges grows rapidly. Whenever the pair domain was large enough, the observed loop-free rank attained the symmetric-tensor bound
\[
\binom{d+1}{2}.
\]
Thus, node ranks $d=4$, $8$, and $12$ produced edge ranks $10$, $36$, and $78$, respectively, except when the nominal pair dimension itself was smaller. At the same time, the proportion of the nominal edge space required by the covariance decreased sharply as $R$ increased. For example, with $d=8$ the rank ratio fell from $36/66\approx0.55$ at $R=12$ to $36/435\approx0.08$ at $R=30$. Figure~\ref{fig:rank} summarizes this pattern by plotting relational covariance rank relative to the nominal number of loop-free pairs. The decreasing ratios show visually why the rank result matters: as the network grows while the node geometry remains low dimensional, the covariance occupies a progressively smaller fraction of the available edge space.

\begin{figure}[htbp]
\centering
\includegraphics[width=0.72\textwidth]{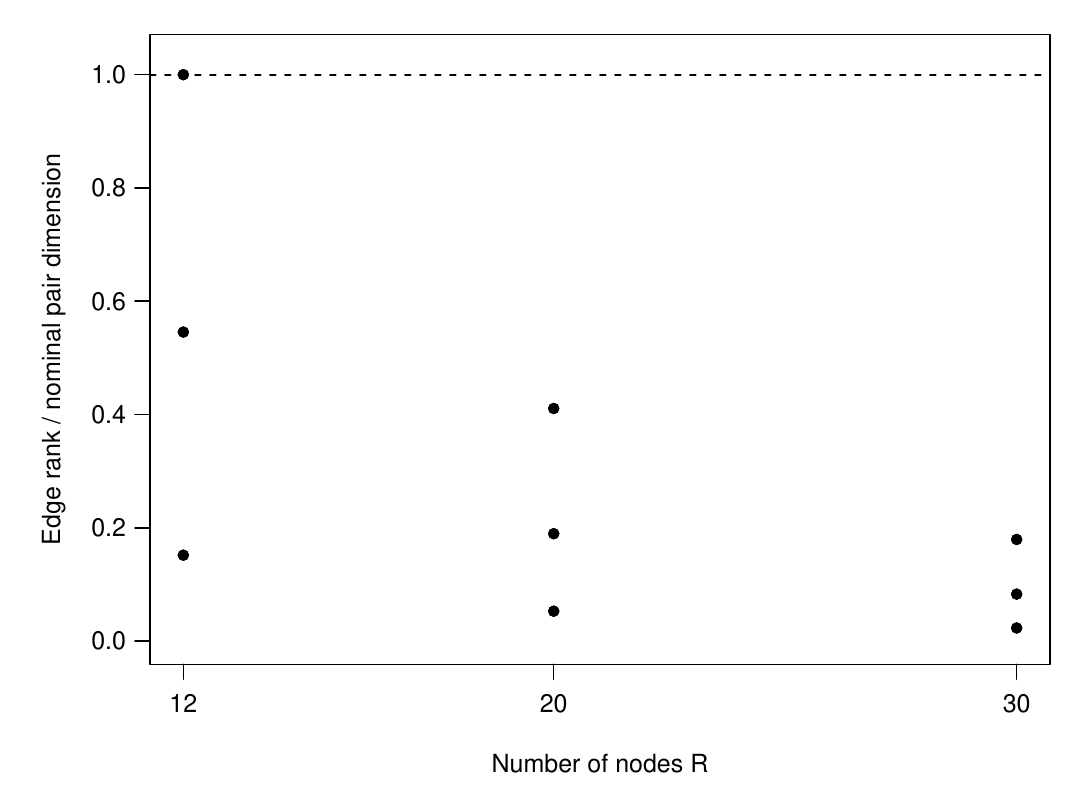}
\caption{Relational covariance rank relative to the nominal loop-free pair dimension. The same node-level rank generates an increasingly compressed relational covariance as the number of nodes grows.}
\label{fig:rank}
\end{figure}

The spectral calculations matched the theory to numerical precision. Across all 27 combinations of network size, target rank, and eigenvalue-decay pattern, the maximum discrepancy between the numerically computed full symmetric-product spectrum and the theoretical product spectrum was only
\[
1.14\times10^{-12}.
\]
The largest numerical violation of the loop-free interlacing inequalities was
\[
3.18\times10^{-14},
\]
which is at the level of floating-point precision. The simulations therefore recover both the established full-product spectrum and the new loop-free spectral constraints essentially exactly.

Removing self-pairs did reduce total covariance trace, but the relative loss became smaller as the network grew. Across the simulated settings, the observed trace fraction removed ranged from approximately 9.6\% to 25.5\%, and every value remained below the normalized theoretical bound $4/(R+1)$. This is useful in practice: the loop-free restriction changes the spectrum, but it does not erase the low-dimensional structure inherited from the node geometry.

\subsection{Risk gains depend on geometric alignment}

The second experiment provides the most direct statistical interpretation of the relational covariance. When the data-generating effect followed the assumed covariance, the structured estimator produced very large reductions in risk. The exact structured-to-unstructured risk ratios were 0.053, 0.085, and 0.107 for $\tau^2=0.25$, $1$, and $4$, respectively. In other words, in these settings the structured estimator used only about 5--11\% of the risk of the unstructured estimator. Monte Carlo estimates agreed almost exactly with the theoretical expressions.

The more revealing result came from moving signal away from the directions favored by the covariance. Figure~\ref{fig:misspec} shows a smooth deterioration in relative performance as signal energy was transferred from the leading eigendirection to a weak direction. The figure therefore visualizes the bias--variance mechanism in Theorem~\ref{thm:risk}: the benefit is greatest when the effect lies in directions that the covariance regards as plausible. For $\tau^2=1$, the exact risk ratio increased from 0.07 under full alignment to 0.39 at 50\% misalignment and 0.71 when the signal lay entirely in the weak direction.

\begin{figure}[htbp]
\centering
\includegraphics[width=0.72\textwidth]{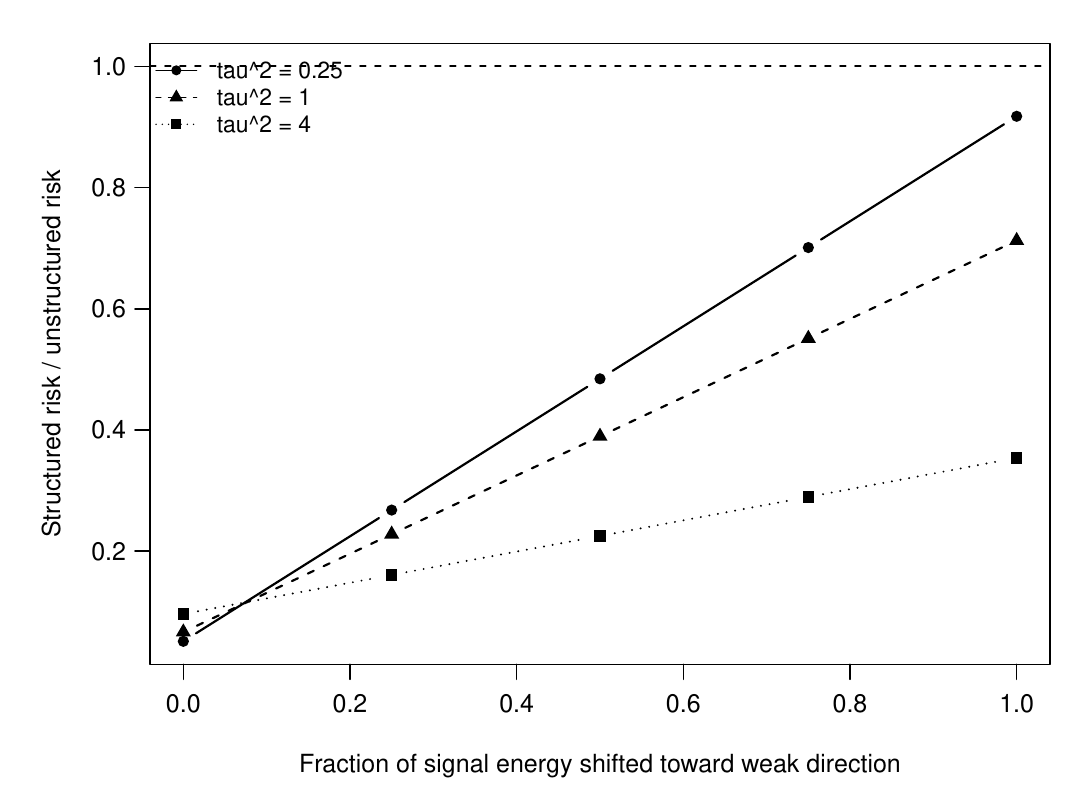}
\caption{Effect of geometric misspecification on conditional risk. Structured borrowing is highly effective when the signal aligns with high-eigenvalue relational directions, but its advantage diminishes as signal energy moves toward weakly supported directions.}
\label{fig:misspec}
\end{figure}

The severe misspecification scenario is important because it shows the other side of structured regularization. When the entire enlarged signal was placed in the weakest supported direction and shrinkage was made deliberately strong, the exact structured-to-unstructured risk ratio rose to
\[
8.79.
\]
Thus the structured estimator was substantially worse than the unstructured estimator in this intentionally adverse setting. Figure~\ref{fig:riskcases} places that adverse result next to the aligned and moderately misspecified cases. The comparison is included to prevent a misleading reading of the Bayes-risk result: structured borrowing can be highly effective, but severe geometric misspecification can reverse the advantage.

\begin{figure}[htbp]
\centering
\includegraphics[width=0.72\textwidth]{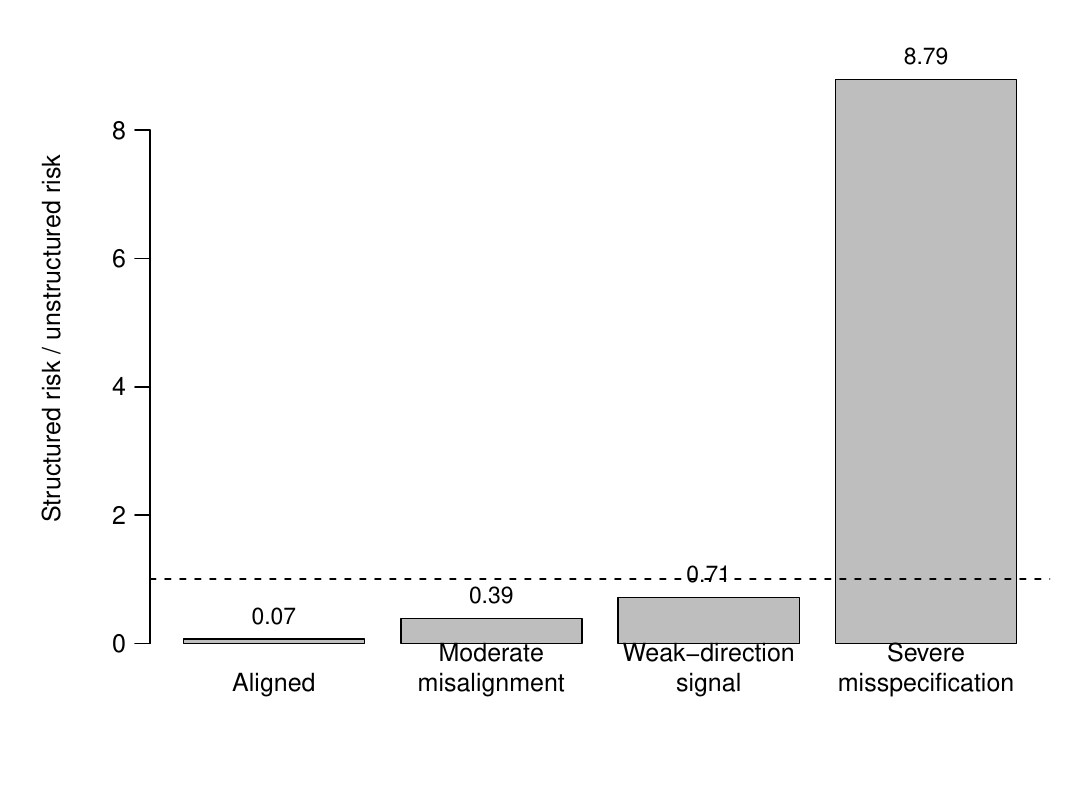}
\caption{Selected risk scenarios. The dashed horizontal line marks equal risk. The first three scenarios illustrate the gradual erosion of the benefit of relational regularization; the final scenario deliberately combines severe geometric misalignment with strong shrinkage and produces clear risk inflation.}
\label{fig:riskcases}
\end{figure}

This pattern is important because it separates model-based benefit from unconditional dominance. The relational covariance is not acting as a generic smoother that is always beneficial. It is translating a scientific geometry into a preferred set of relational directions. When that geometry is informative, the gains can be substantial. When it is badly wrong, the risk calculation reveals the cost. That is precisely the behavior a useful structured method should make transparent.

\subsection{Low-rank approximation preserves inference progressively}

The third experiment examines whether the covariance can be compressed without losing the inferential behavior that motivated it. The discrepancy decreased steadily as additional covariance directions were retained. As retained rank increased, covariance trace increased smoothly and estimator discrepancy declined monotonically. At rank 39, approximately 97.1\% of covariance trace was retained and the mean estimator discrepancy was 1.56. At rank 46, the retained trace rose to 99.1\% and the discrepancy fell to 0.70. By rank 54, more than 99.96\% of trace was retained and mean discrepancy was only 0.07.

Figure~\ref{fig:trunc} compares the observed estimator discrepancy with the theorem-based bound. Its purpose is not only to verify the inequality, but to show how an inferential tolerance can guide rank selection more directly than a trace threshold. The bound is conservative, as expected for an operator-norm guarantee, but it decreases in the correct direction and remains valid throughout. The largest observed ratio of actual error to the bound was below one.

\begin{figure}[htbp]
\centering
\includegraphics[width=0.72\textwidth]{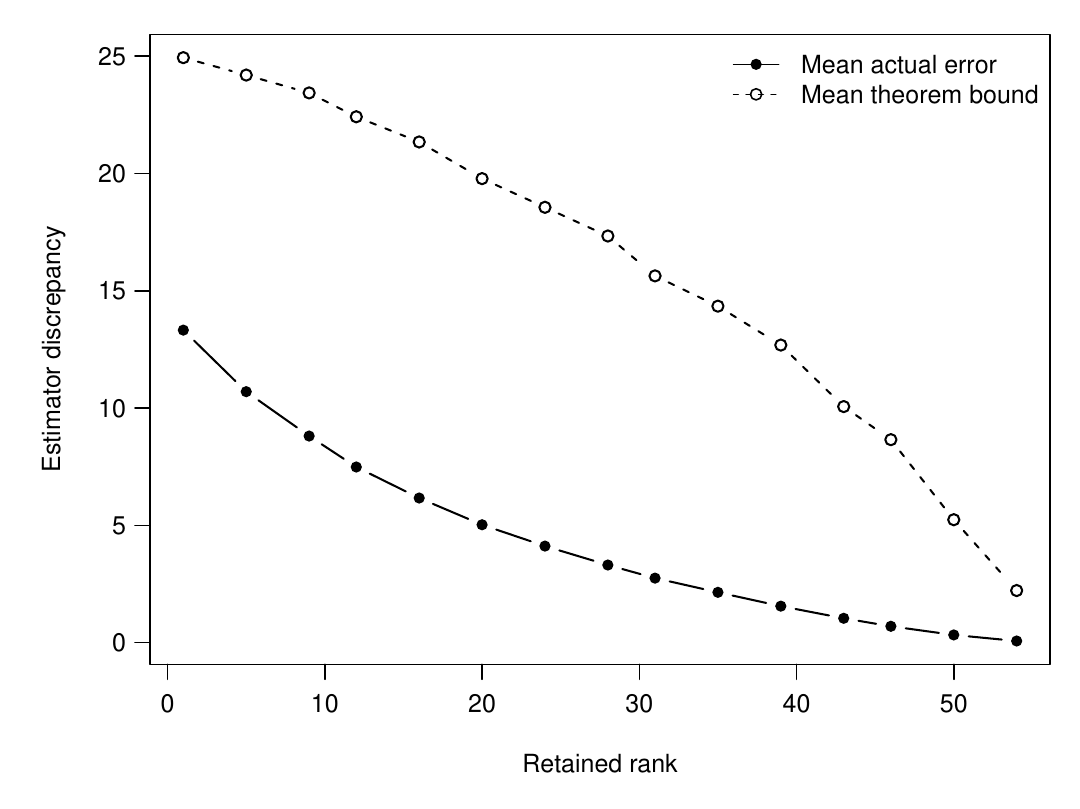}
\caption{Inferential error under spectral truncation. Both the observed discrepancy between the full and truncated estimators and the theoretical upper bound decrease as more relational eigencomponents are retained.}
\label{fig:trunc}
\end{figure}

This result is useful computationally because it gives a more meaningful criterion than choosing rank solely by percentage of variance explained. The approximation can instead be evaluated on the scale of the estimator itself.

\subsection{Node-geometry error propagates smoothly to the pair covariance}

The final experiment evaluates the new operator perturbation result. As the node features were increasingly perturbed, the induced pair-covariance error increased smoothly and remained below the theoretical bound at every perturbation level. Mean pair-covariance operator error rose from 0.224 at perturbation level 0.01 to 4.396 at level 0.20, while the corresponding theorem bounds were 0.330 and 6.575. Figure~\ref{fig:perturb} shows that the bound tracks the scale of the observed perturbation rather than merely providing a vacuous worst-case inequality. This is the numerical counterpart of the scientific stability question: moderate uncertainty in node representation produces a controlled, predictable change in the induced relational covariance.

\begin{figure}[htbp]
\centering
\includegraphics[width=0.72\textwidth]{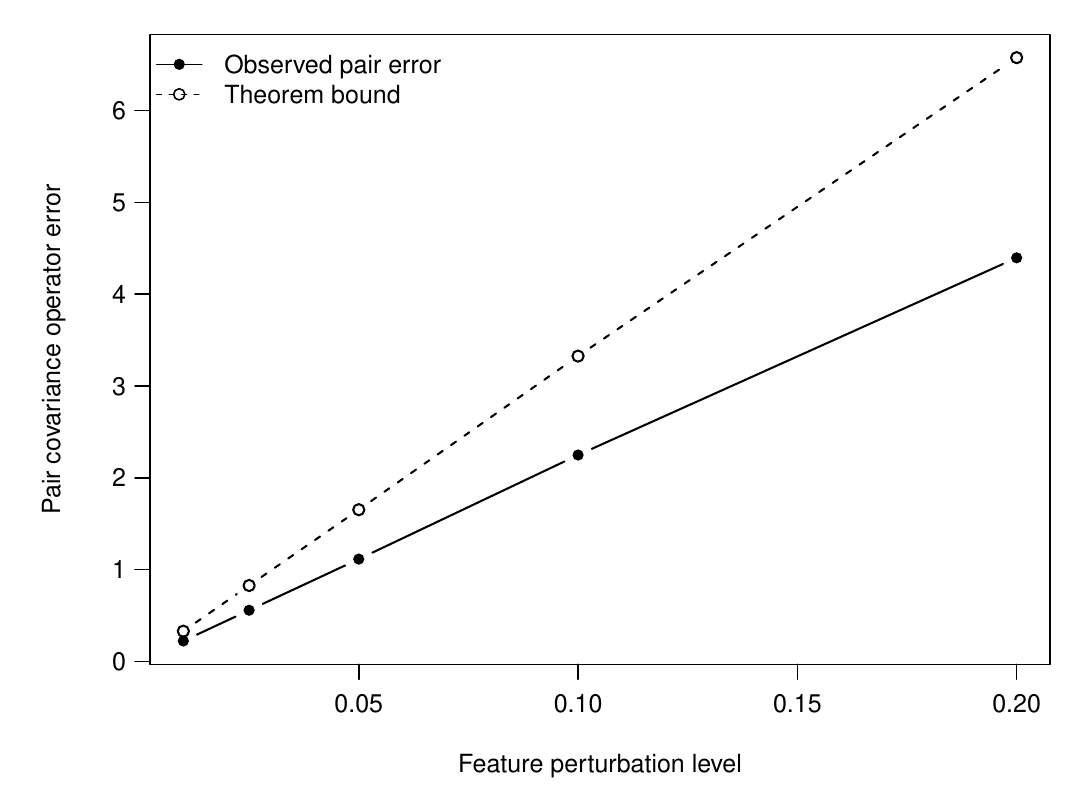}
\caption{Propagation of node-geometry perturbation to the loop-free pair covariance. The observed operator error grows smoothly with perturbation level and remains below the bound from Theorem~\ref{thm:perturb}.}
\label{fig:perturb}
\end{figure}

The estimator-level bound was more conservative than the covariance-level bound, which is unsurprising because it compounds a worst-case operator inequality with the norm of the observed vector. The covariance-level result is therefore the more informative practical diagnostic. Even so, the full chain from node geometry to pair covariance to estimator remained stable throughout the experiment.

Taken together, the simulations serve two purposes. They verify that the finite loop-free calculations behave as the theory predicts, and they show why the inferential qualifications matter. Low-dimensional geometry can induce strong compression and substantial risk reduction, but neither is automatically sufficient: misspecified geometry can increase risk, and a high percentage of retained trace can still leave a non-negligible estimator discrepancy. The perturbation experiment completes this picture by showing that uncertainty in the node representation propagates in a controlled, quantitatively bounded manner.

\FloatBarrier

\section{Application: Relational Covariance in Autism Neuroimaging}
\label{sec:application}

\subsection{Scientific question and data}

Autism provides a compelling test of relational covariance because the scientific question concerns how brain regions work together. Brain function emerges from coordinated activity across distributed systems, and autism is biologically heterogeneous. Meaningful alterations may therefore appear as modest changes across many connections rather than as one dominant abnormality. Statistically, this is difficult: the connectome is high dimensional, individual edge signals can be weak, and global summaries can hide the local relationships that carry biological information.

The Autism Brain Imaging Data Exchange (ABIDE I) provides a useful setting for examining this problem \citep{DiMartino2014}. ABIDE combines resting-state functional magnetic resonance imaging with phenotypic information collected across multiple imaging sites. In a resting-state scan, participants are not asked to perform a specific cognitive task. Instead, the scanner records spontaneous fluctuations in the blood-oxygen-level-dependent signal while the participant is at rest. Regions whose signals fluctuate together are interpreted as exhibiting functional connectivity. This provides a way to study the organization of large-scale brain systems without restricting the analysis to a single task or stimulus.

For the present application, the processed cohort contains 792 participants from 20 imaging sites. Brain activity is represented using the 116-region Automated Anatomical Labeling atlas \citep{TzourioMazoyer2002}. A regional time series is obtained for each atlas region, and each participant is represented by pairwise functional connectivity across the 116 regions. Because the network is undirected, the connection between regions $r$ and $s$ is the same as the connection between $s$ and $r$. Excluding self-connections gives
\[
q=\binom{116}{2}=6{,}670
\]
unique functional connections for each participant.

This representation illustrates the central methodological issue in the paper. The observed outcome is indexed by \emph{relationships} between brain regions, while much of the biological information that can guide inference is available for the regions themselves. Anatomical location is defined at the region level. Population-level functional organization is also naturally summarized through region-level coordinates or embeddings. The statistical challenge is therefore to use that information to determine which of the 6,670 connections should be regarded as related, without treating the connections as independent and without collapsing them into a small number of global network summaries.

The role of the proposed covariance is to provide that bridge. Connections involving anatomically or functionally similar regions can borrow information from one another because their endpoints are biologically related. At the same time, inference remains defined on individual region pairs. This is particularly relevant in autism, where distributed alterations may span several systems and where a model that borrows information only when supported by brain organization may be more informative than either independent edgewise analysis or aggressive dimension reduction.

The goal of the application is not to test an ASD group contrast. It is to ask whether the new theory changes how a real connectome should be represented and regularized. We examine the effective covariance dimension of the 6,670-edge domain and the effect of removing self-pairs. We then study whether spectral compression preserves inference and how uncertainty in the brain geometry propagates to the pair covariance.

\subsection{Biological organization at the region level}

Two complementary forms of region-level organization are used. The first is anatomical. Let $\bm a_r$ denote the three-dimensional centroid of AAL region $r$. Anatomical similarity is defined by
\begin{equation}
\label{eq:abide_anat}
\kappa_A(r,u)
=
\exp\left\{
-\frac{\|\bm a_r-\bm a_u\|^2}{2\rho_A^2}
\right\},
\qquad
\rho_A=31.0.
\end{equation}
This representation reflects physical organization of the brain. Regions that occupy similar anatomical locations are assigned greater similarity, and connections formed by anatomically similar endpoints can therefore receive stronger covariance.

The second representation captures functional organization that is not reducible to physical distance. A diagnosis-blind population functional graph is used to obtain a 12-dimensional spectral embedding $\bm g_r$ for each region, using a normalized graph operator and standard spectral graph principles \citep{vonLuxburg2007}. Functional similarity is then defined by
\begin{equation}
\label{eq:abide_func}
\kappa_G(r,u)
=
\exp\left\{
-\frac{\|\bm g_r-\bm g_u\|^2}{2\rho_G^2}
\right\},
\qquad
\rho_G=0.50.
\end{equation}
Because this representation is diagnosis-blind, it describes population-level functional organization rather than using the ASD contrast itself to define the covariance.

Each region-level similarity is lifted to the unordered-connection domain through \eqref{eq:lift}. We also use a combined representation that requires support from both anatomical and functional organization. These geometries offer different biological notions of which connections should be allowed to share information.

Figure~\ref{fig:abidebrain} shows the 116 AAL region centroids. The figure makes the change of domain explicit. The points provide the biological information used to construct one component of the covariance, but the statistical domain consists of the 6,670 connections formed from them.

\begin{figure}[htbp]
\centering
\includegraphics[width=0.94\textwidth]{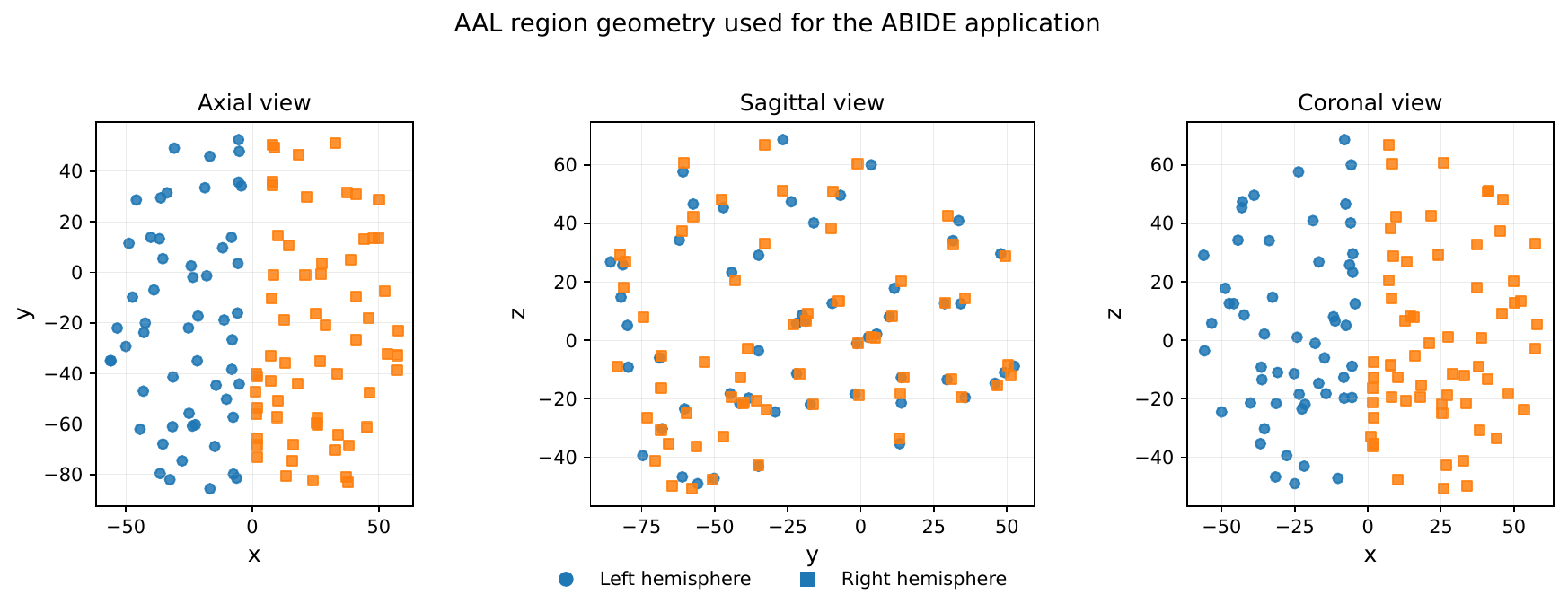}
\caption{AAL region centroids defining the anatomical geometry in axial, sagittal, and coronal projections. The application is indexed by the 6,670 unordered region pairs formed from these 116 regions rather than by the regions individually.}
\label{fig:abidebrain}
\end{figure}

\subsection{A high-dimensional connectome with low-dimensional relational structure}

We begin by asking whether the nominal size of the connectome reflects its actual covariance complexity. If the 6,670 connections behaved as unrelated coordinates, one would expect the effective dimension to remain close to the edge count. If instead the connectome is organized by a smaller number of biologically meaningful relational patterns, the induced covariance should concentrate in far fewer directions.

Table~\ref{tab:abidespectral} shows strong evidence of the latter. The numerical ranks of the anatomical, functional, and combined relational representations are 741, 136, and 256. These values are already much smaller than 6,670. More importantly, the covariance mass is concentrated even more strongly than the numerical ranks suggest.

The functional geometry is especially concentrated. Only 14 directions explain 90\% of covariance trace and 51 directions explain 95\%. Thus 95\% of the functional relational covariance is represented by only
\[
\frac{51}{6670}\times100=0.76\%
\]
of the nominal edge coordinates. The anatomical geometry is less concentrated but remains strongly structured: 360 directions explain 95\% of trace, corresponding to 5.40\% of the edge domain. The combined representation reaches 95\% with 195 directions, or 2.92\% of the full relation space.

\begin{table}[htbp]
\centering
\caption{Spectral complexity of the observed ABIDE relational geometries. Percentages in parentheses are fractions of the nominal 6,670-dimensional loop-free edge domain.}
\label{tab:abidespectral}
\small
\begin{tabular}{lrrrrrr}
\toprule
Geometry & Numerical & Effective & Entropy & 90\% trace & 95\% trace & 99\% trace\\
& rank & rank & rank & rank & rank & rank\\
\midrule
Anatomical & 741 & 17.86 & 205.99 & 255 (3.82\%) & 360 (5.40\%) & 556 (8.34\%)\\
Functional & 136 & 1.69 & 8.35 & 14 (0.21\%) & 51 (0.76\%) & 88 (1.32\%)\\
Combined & 256 & 23.72 & 155.42 & 158 (2.37\%) & 195 (2.92\%) & 240 (3.60\%)\\
\bottomrule
\end{tabular}
\end{table}

Biologically, this means that thousands of functional connections are not being organized as thousands of unrelated quantities. The anatomical and functional descriptions of the brain induce a much smaller set of dominant relational patterns. That structure is important in autism research because a distributed signal that is weak at individual edges can still be coherent when viewed across related connections. The covariance provides a mechanism for recognizing that coherence without abandoning connection-level interpretation.

Figure~\ref{fig:abidespectrum} shows the different spectral signatures of the three biological geometries. The functional covariance concentrates very rapidly in its leading directions. Anatomical organization is more diffuse. The combined representation lies between the two. These differences show that the choice of biological geometry is not simply a computational device; it changes the directions along which the model expects connections to vary together.

\begin{figure}[htbp]
\centering
\includegraphics[width=0.92\textwidth]{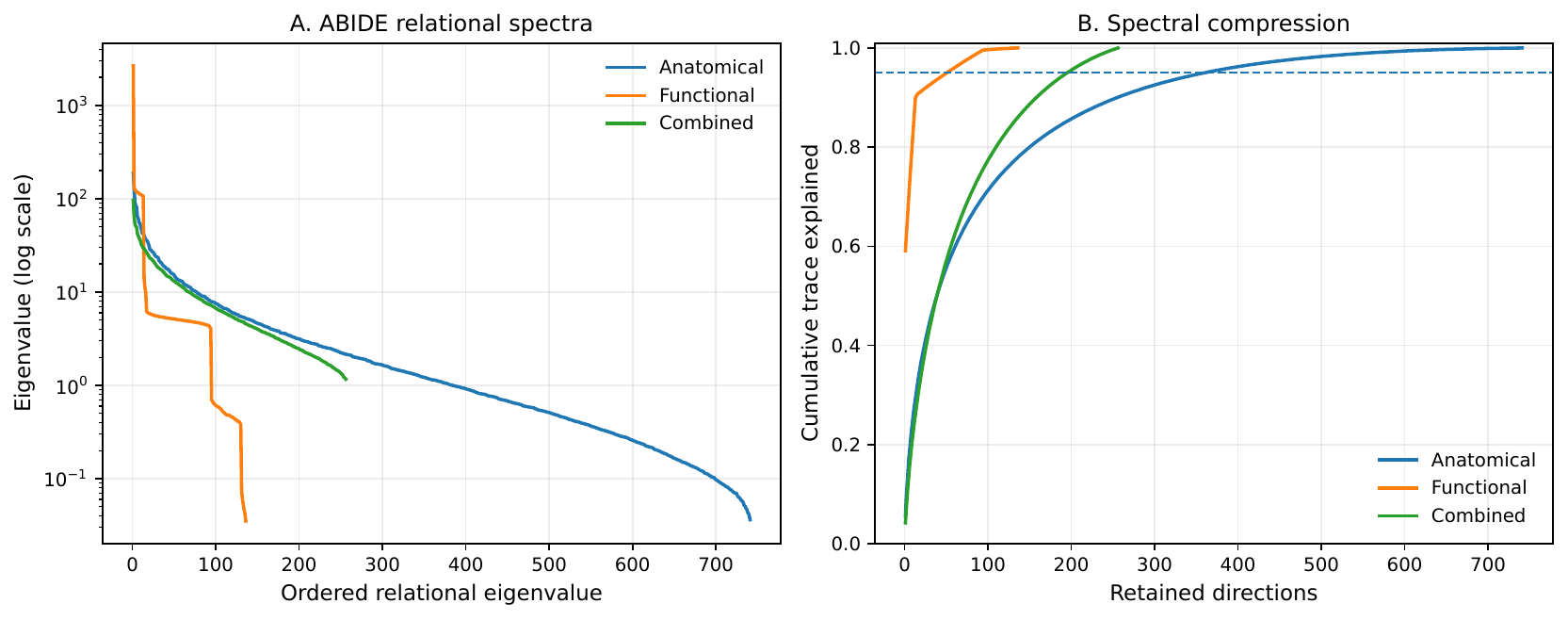}
\caption{Observed ABIDE relational spectra. Left: ordered positive eigenvalues on a logarithmic scale. Right: cumulative covariance trace explained by retained relational directions; the horizontal line marks 95\%.}
\label{fig:abidespectrum}
\end{figure}

The same idea can be visualized directly on the brain. Figure~\ref{fig:abidemodes} displays the 30 largest-magnitude connection loadings in the first eigenvector of each relational covariance. The anatomical mode emphasizes a more spatially concentrated pattern, while the functional mode connects a broader set of region pairs. The combined mode reflects a distinct organization informed by both sources. These are covariance modes, not ASD effect maps. They show the connection patterns along which the model is most willing to share information.

\begin{figure}[htbp]
\centering
\includegraphics[width=0.94\textwidth]{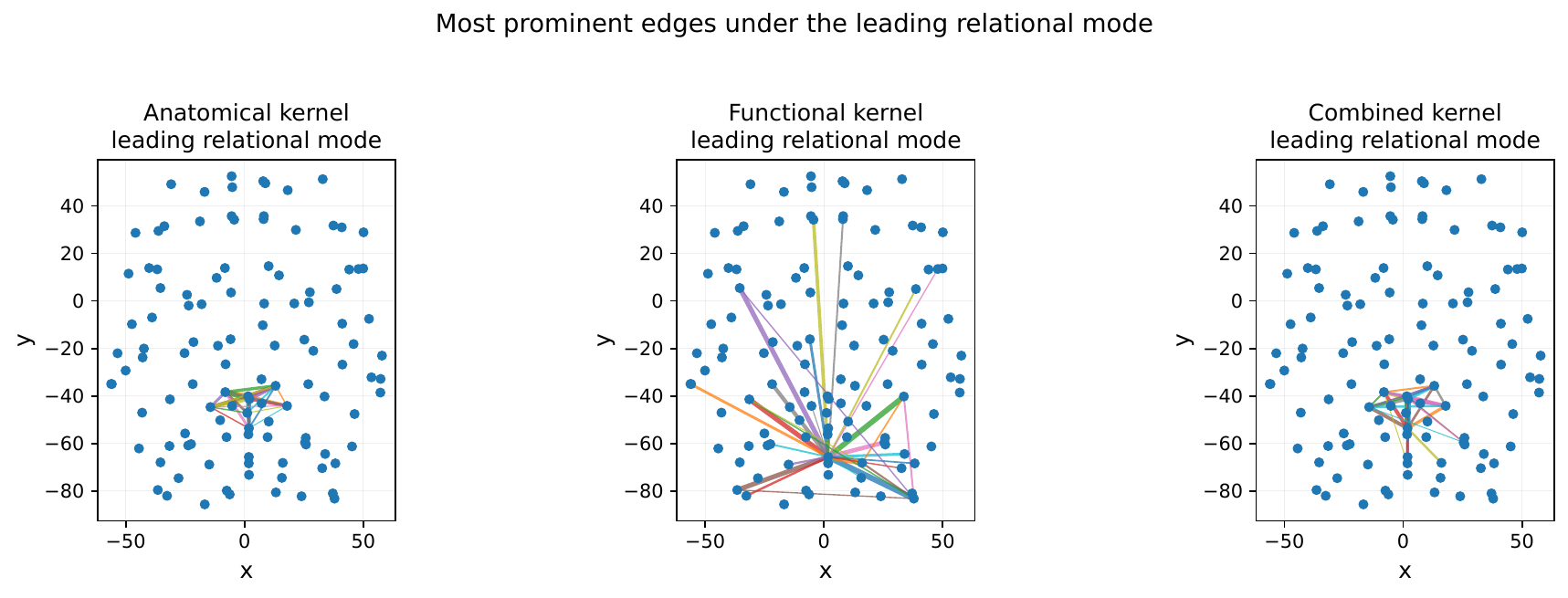}
\caption{Leading relational modes displayed on the observed AAL brain geometry. Each panel shows the 30 edge coordinates with largest absolute loading in the first eigenvector of the anatomical, functional, or combined relational kernel. The figure visualizes covariance geometry rather than an ASD contrast.}
\label{fig:abidemodes}
\end{figure}

\subsection{Why removing self-pairs matters}

Functional connectivity is defined between distinct brain regions. Region-to-itself connections are therefore excluded from the observed connectome. The clean symmetric-product theory, however, is naturally expressed on the larger domain that includes those self-pairs. For the present atlas this full domain contains
\[
q_+=\binom{117}{2}=6{,}786
\]
unordered pairs. The observed network removes exactly 116 coordinates.

Theorem~\ref{thm:loopfree} allows us to quantify what this biologically necessary restriction does to the covariance. Table~\ref{tab:loopfreeapp} shows that, for the anatomical geometry, the full symmetric-pair trace is 7,224.71. Removing the 116 self-pairs reduces the trace by 116, or 1.606\%. The generic normalized bound is
\[
\frac{4}{R+1}
=
\frac{4}{117}
=
0.03419,
\]
which corresponds to 3.419\%. The observed loss is therefore well below the theoretical limit. Among the first 150 loop-free eigenvalues examined, no interlacing violation was observed to reported numerical precision.

\begin{table}[htbp]
\centering
\caption{Observed effect of the loop-free restriction for the anatomical ABIDE geometry.}
\label{tab:loopfreeapp}
\small
\begin{tabular}{lr}
\toprule
Quantity & Observed value\\
\midrule
Full symmetric-pair dimension & 6,786\\
Loop-free dimension & 6,670\\
Self-pairs removed & 116\\
Full symmetric-pair trace & 7,224.71\\
Trace removed & 116.00\\
Trace-loss fraction & 1.606\%\\
Generic bound $4/(R+1)$ & 3.419\%\\
Leading eigenvalues checked & 150\\
Maximum interlacing violation & $0.000\times10^{0}$\\
\bottomrule
\end{tabular}
\end{table}

Figure~\ref{fig:loopfreeapp} shows the same result spectrally. The loop-free eigenvalues remain between the upper full-product spectrum and the lower interlacing envelope. The result matters because it justifies using the full symmetric product as a mathematical reference while still respecting the biological domain actually analyzed.

\begin{figure}[htbp]
\centering
\includegraphics[width=0.72\textwidth]{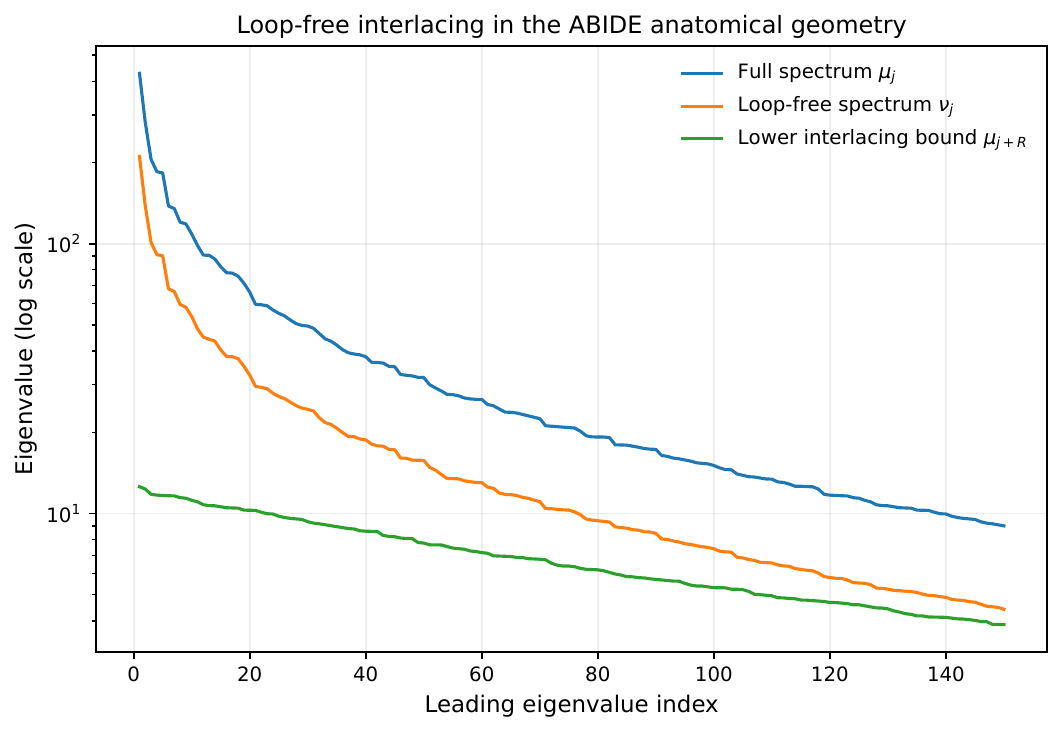}
\caption{Loop-free spectral interlacing for the observed anatomical ABIDE geometry. The loop-free eigenvalues $\nu_j$ lie between the full symmetric-product eigenvalues $\mu_j$ and the lower interlacing bound $\mu_{j+R}$.}
\label{fig:loopfreeapp}
\end{figure}

\subsection{When spectral compression is sufficient for inference}

The strong reduction in effective dimension suggests an obvious computational strategy: retain only the dominant relational directions. For a connectome of this size, that can substantially reduce storage and computation. The scientific question, however, is not only how much covariance is preserved. It is whether the resulting estimator remains close enough to the one obtained from the full covariance.

Theorem~\ref{thm:exacttrunc} makes that distinction measurable. Using the canonical scaling $\tau^2=v=1$, a 95\%-trace approximation does not necessarily produce a small change in the shrinkage operator. For the anatomical geometry, retaining 360 directions leaves an exact operator discrepancy of 0.533. The functional rank-51 approximation leaves a discrepancy of 0.836, and the combined rank-195 approximation leaves a discrepancy of 0.718.

More directions are needed when the goal is to preserve inferential behavior. An operator error below 0.10 requires 690 anatomical directions and 130 functional directions. For the combined representation, all 256 available directions are required. Table~\ref{tab:truncapp} summarizes this difference.

\begin{table}[htbp]
\centering
\caption{Spectral truncation versus inferential approximation under the canonical $\tau^2=v=1$ scaling. The operator errors are exact consequences of Theorem~\ref{thm:exacttrunc}.}
\label{tab:truncapp}
\small
\begin{tabular}{lrrrrrr}
\toprule
Geometry & 95\% trace & Error at & Rank for & Rank for & Rank for\\
& rank & 95\% rank & error $<0.10$ & error $<0.05$ & error $<0.01$\\
\midrule
Anatomical & 360 & 0.533 & 690 & 733 & 741\\
Functional & 51 & 0.836 & 130 & 132 & 136\\
Combined & 195 & 0.718 & 256 & 256 & 256\\
\bottomrule
\end{tabular}
\end{table}

Figure~\ref{fig:truncapp} shows the exact discrepancy across all retained ranks. The horizontal lines provide direct inferential tolerances rather than a variance-explained threshold.

\begin{figure}[htbp]
\centering
\includegraphics[width=0.72\textwidth]{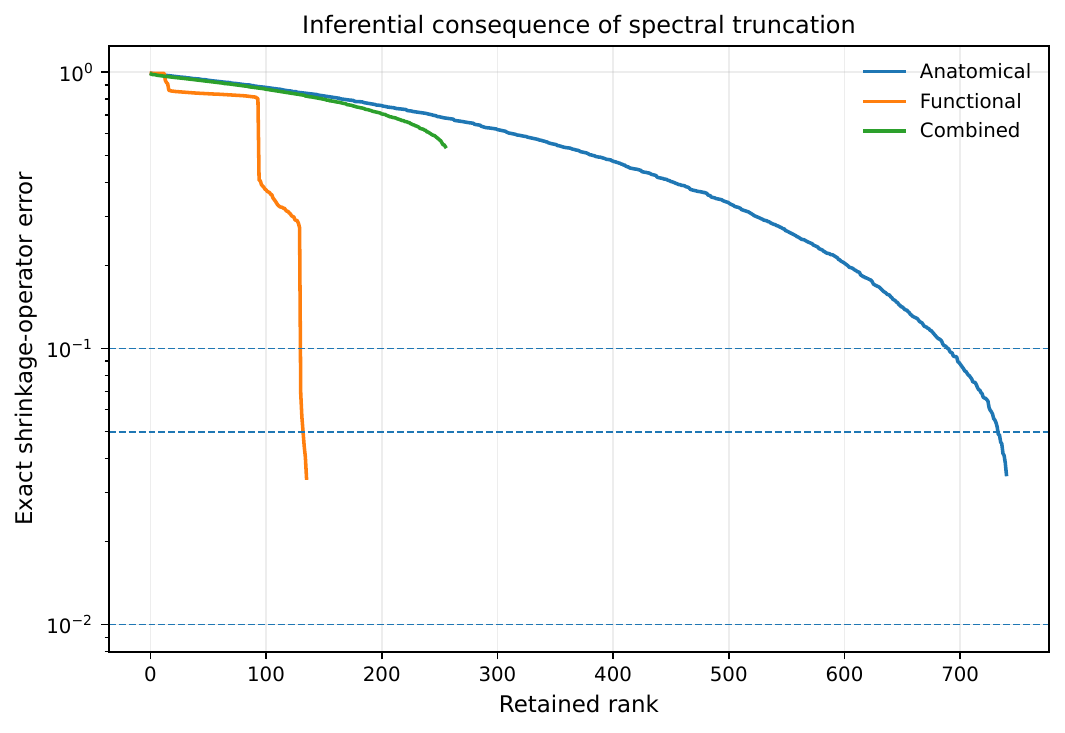}
\caption{Exact inferential consequence of spectral truncation in the ABIDE relational geometries for $\tau^2=v=1$. Horizontal lines mark operator errors of 0.10, 0.05, and 0.01. High cumulative trace alone does not guarantee a small change in the shrinkage operator.}
\label{fig:truncapp}
\end{figure}

This distinction matters in a disease study because the covariance determines how strongly connection-level effects are regularized. An approximation can reproduce most of the covariance variation and still alter the inferential map appreciably. The proposed result therefore connects computational reduction to the quantity that ultimately matters for analysis: the behavior of the estimator.

\subsection{Robustness to uncertainty in anatomical representation}

The biological geometry used to organize the connectome is not known without error. Atlas coordinates are representations of regions rather than exact biological locations, and data-derived embeddings are subject to sampling variability. It is therefore important to know whether a modest perturbation at the region level can produce an unstable change across thousands of pair coordinates.

We examine this question using the observed AAL centroids. Independent Gaussian perturbations are added to each anatomical coordinate with standard deviations 0.5, 1, 2, and 5 atlas-coordinate units. At each level, 100 perturbed node geometries are generated and lifted to the loop-free pair domain.

Table~\ref{tab:perturbapp} reports the resulting operator errors. Mean pair-covariance error increases from 3.88 at perturbation SD 0.5 to 38.06 at SD 5. The corresponding theorem bound increases from 5.86 to 57.06. The observed error remains a stable fraction of the bound: the mean actual-to-bound ratio lies between 0.655 and 0.667 across all perturbation levels.

\begin{table}[htbp]
\centering
\caption{Propagation of perturbations in the observed anatomical node geometry to the loop-free pair covariance. Results are based on 100 perturbations at each level.}
\label{tab:perturbapp}
\small
\begin{tabular}{rrrrr}
\toprule
Centroid & Mean node & Mean pair & Mean theorem & Mean actual /\\
SD & error & error & bound & bound\\
\midrule
0.5 & 0.283 & 3.880 & 5.857 & 0.662\\
1.0 & 0.559 & 7.587 & 11.565 & 0.655\\
2.0 & 1.129 & 15.474 & 23.324 & 0.663\\
5.0 & 2.780 & 38.058 & 57.056 & 0.667\\
\bottomrule
\end{tabular}
\end{table}

Figure~\ref{fig:perturbapp} shows that the observed pair-covariance error and the theorem bound increase in parallel as the anatomical perturbation grows. The bound therefore responds to the magnitude of uncertainty rather than acting only as a loose worst-case statement.

\begin{figure}[htbp]
\centering
\includegraphics[width=0.72\textwidth]{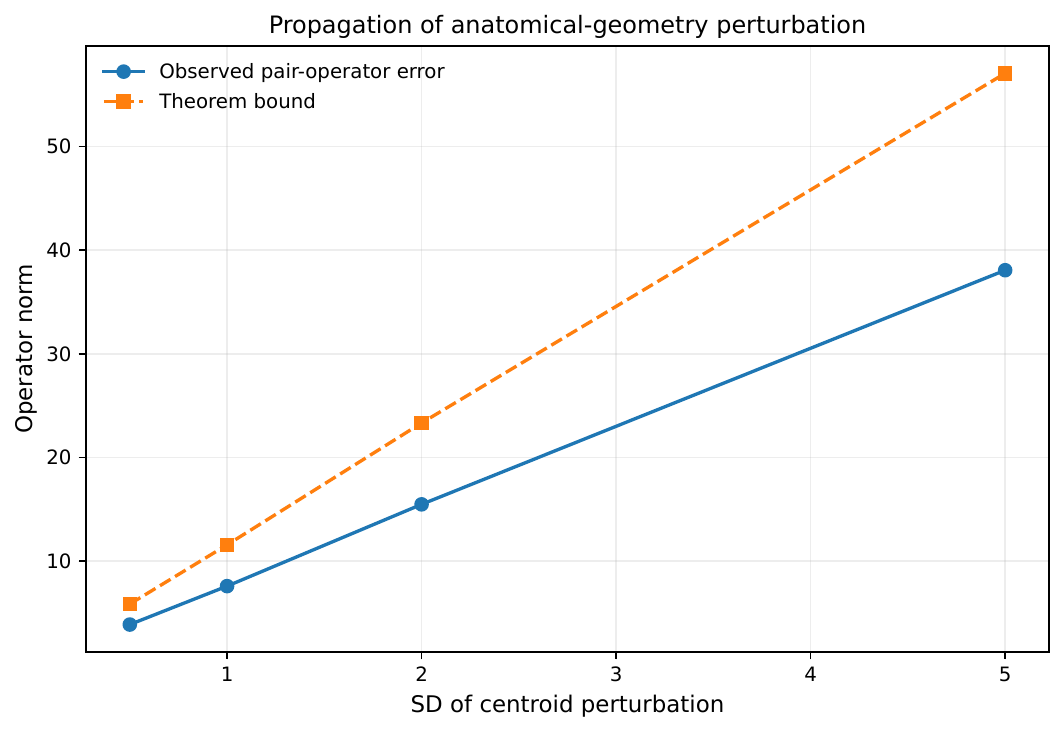}
\caption{Propagation of perturbations in the anatomical node geometry. The observed pair-covariance operator error grows approximately linearly with the perturbation scale and remains below the theorem bound throughout.}
\label{fig:perturbapp}
\end{figure}

\subsection{What the application shows}

The ABIDE results show why relational covariance is useful in a biological network. The 6,670 possible connections do not translate into 6,670 equally important covariance directions. Anatomy and population-level functional organization induce much smaller relational structures. That finding matters for autism research because distributed changes can be difficult to identify one edge at a time, yet collapsing the connectome to a global summary can erase the specific relationships through which systems differ.

The proposed framework preserves that connection-level interpretation while giving the analysis a principled way to share information. The loop-free result respects the biological domain actually observed. The truncation result shows how much compression can be used without relying only on variance explained. The perturbation analysis shows how uncertainty in the region representation can be carried into the covariance rather than ignored. Together, these results demonstrate the practical value of treating the connectome as a structured relational domain rather than as an unorganized vector of edges.

\FloatBarrier

\section{Discussion}
\label{sec:discussion}

The paper is motivated by a simple mismatch that appears in many network problems. Scientific information often describes individual objects, while inference concerns the relationships among them. A long vector of edges does not resolve that mismatch. The covariance still has to say which relationships should inform one another. The framework developed here provides that bridge by carrying object-level geometry to unordered relationships while preserving both endpoints.

Symmetric pairwise kernels provide the algebraic starting point \citep{Pahikkala2015,Stock2018,Viljanen2022}. The innovation here is the covariance and inferential theory that becomes necessary when the statistical process lives on a loop-free relation domain. The exact product spectrum belongs to the full symmetric pair space, yet real undirected networks remove self-pairs. Theorem~\ref{thm:loopfree} shows how much of the spectral structure survives that restriction. In both the simulations and ABIDE geometry, the resulting interlacing and trace bounds are informative rather than merely formal.

The relational spectrum also acquires a direct statistical meaning once the covariance is used for regularization. Large-eigenvalue directions represent patterns that the geometry supports and are shrunk less. Weak directions are shrunk more. The risk calculation shows why this can produce large gains when the geometry agrees with the signal, while also exposing the cost of strong misspecification. That qualification is important: scientific geometry is a modeling assumption that can be evaluated, not an automatic source of improved estimation.

A similar distinction appears in low-rank computation. The ABIDE results show substantial spectral compression, but they also show that variance explained and inferential fidelity answer different questions. The exact truncation result makes it possible to choose rank according to an estimator tolerance rather than a conventional trace threshold. This becomes especially valuable in large networks, where aggressive compression is attractive but the final goal remains inference.

The perturbation result addresses another practical reality. Anatomical coordinates and learned embeddings are not perfectly known. Error in that representation enters before the pair covariance is formed. By following the perturbation from the node kernel to the loop-free covariance and then to the estimator, the theory provides a direct way to assess whether scientific conclusions are sensitive to the geometry used to organize the network. The ABIDE perturbation experiment suggests that the bound remains informative over a meaningful range of anatomical perturbations.

There are several promising directions for further work. A probabilistic treatment of estimated geometry could connect these deterministic bounds to convergence theory for learned embeddings and covariance estimators \citep{BergerHolzmann2024}. Non-Gaussian outcomes and time-varying relational processes would extend the inferential setting. The higher-order construction in Section~\ref{sec:higher} also suggests analogous theory for relations involving more than two objects.

The broader point is that unordered relationships deserve to be treated as a statistical domain in their own right. When scientifically meaningful structure is available for the constituent objects, it can guide dependence among their relationships without collapsing the network or imposing an artificial ordering. The results here make that idea usable for inference by showing what survives on the loop-free domain, how the covariance regularizes estimation, and how approximation and geometric uncertainty affect the analysis.

\section{Conclusions}
\label{sec:conclusions}

This paper develops covariance theory for stochastic processes indexed by loop-free unordered relationships. Object-level scientific structure is transferred to the relation domain through a symmetric construction. The theory then follows the covariance through the steps that matter for inference, from loop removal and regularization to approximation and uncertainty in the underlying geometry.

The results show that a large relational domain can have a much smaller effective covariance dimension, creating opportunities for structured borrowing and computation. They also show why that opportunity must be used carefully. Geometric alignment determines the value of regularization, and covariance trace alone does not determine inferential accuracy. The ABIDE application demonstrates that these distinctions are visible in a real connectome. More broadly, the framework provides a principled route from scientific information about individual objects to stable and interpretable inference about the relationships they form.

\section*{Author Contributions}

The author conceived the study and developed the theory and methodology. She performed the analyses, prepared the visualizations, and wrote and revised the manuscript.

\section*{Funding}

This research received no specific external funding.

\section*{Institutional Review Board Statement}

Not applicable. This study uses only publicly available, de-identified secondary data from ABIDE I.

\section*{Informed Consent Statement}

Not applicable to the present secondary analysis of publicly available, de-identified data.

\section*{Data Availability Statement}

The neuroimaging data analyzed in this study were obtained from the Autism Brain Imaging Data Exchange I (ABIDE I), an existing third-party data resource available through the 1000 Functional Connectomes Project/International Neuroimaging Data-sharing Initiative (INDI) at \url{https://fcon_1000.projects.nitrc.org/indi/abide/abide_I.html}. Access to the imaging data requires registration with NITRC and INDI and is subject to the repository's data-use requirements. The present study did not generate new human-subject neuroimaging data.

\section*{Acknowledgments}

The author thanks the ABIDE investigators, participating sites, and study participants for making this neuroimaging resource available to the research community.

\section*{Conflicts of Interest}

The author declares no conflict of interest.

\bibliographystyle{apalike}
\bibliography{references}

\end{document}